%% file: main.tex
\documentclass[10pt,letterpaper]{IEEEtran}
\usepackage{xcolor}
\usepackage{acronym}
\usepackage{bm}
\usepackage{amsmath}
\usepackage{amsfonts}
\usepackage{comment}
\usepackage{bbm}
\usepackage{url,cite}
\usepackage{mathtools}
\usepackage[utf8]{inputenc}
\usepackage[english]{babel}
\usepackage{graphicx}
\usepackage{tikz}
\usepackage{pgfplots}
\usepackage{amsthm}
\usepackage{mathtools}

\newcommand{\Corr}[1]{#1}

\newtheorem{theorem}{Theorem}
\newtheorem{lemma}{Lemma}
\newtheorem{proposition}{Proposition}
\newtheorem{conjecture}{Conjecture}
\newtheorem{corollary}{Corollary}
\input{acros}

\begin{document}
\bstctlcite{IEEEexample:BSTcontrol}

%\title{RIS Profile Encoding under Limited Phase Resolution for Interpath Interference Elimination}
\title{Multi-RIS Discrete-Phase Encoding for Interpath-Interference-Free Channel Estimation}

%\title{Interpath %Interference %Elimination in %Multi-RIS Systems for %with Applications in 
%Channel Estimation and %Localization}
\author{Kamran Keykhosravi,~\IEEEmembership{Member,~IEEE,}
        and Henk Wymeersch,~\IEEEmembership{Senior Member,~IEEE}
	\thanks{ This work was supported, in part, by the Swedish Research Council under grant 2018-03701 and the EU H2020 RISE-6G project under grant 101017011.}
\thanks{The authors are with Chalmers University of Technology, Department of Electrical Engineering, Sweden.}}

%\markboth{Journal of \LaTeX\ Class Files,~Vol.~14, No.~8, August~2015}%
%{Shell \MakeLowercase{\textit{et al.}}: Bare Demo of IEEEtran.cls for IEEE Journals}
\maketitle

\begin{abstract}
Reconfigurable intelligent surfaces (RISs) are one of the foremost technological enablers of future wireless systems. They improve  communication and localization by providing a strong non-line-of-sight path to the receiver. In this paper, we propose a pilot transmission method to enable the receiver to separate signals arriving from different RISs and from the uncontrolled multipath. This facilitates channel estimation and localization, as the channel or its geometric parameters can be estimated for each path separately. Our method is based on designing temporal phase profiles that are orthogonal across RISs without affecting the RIS beamforming capabilities. We take into consideration the limited resolution of the RIS phase shifters and show that in the presence of this practical limitation, orthogonal phase profiles can be designed based on Butson-type Hadamard matrices. For a localization scenario, we show that with our proposed method the estimator can attain the theoretical lower bound even with one-bit RIS phase resolution. 
\end{abstract}
\IEEEpeerreviewmaketitle
\vspace{-5mm}
\section{Introduction}

\IEEEPARstart{O}{ne} of the most-promising emerging technologies in 6G communication systems are \acp{ris}, which enable the 
control and optimization of the propagation channel, thereby enabling or boosting radio communication, localization, and sensing \cite{bjornson2021reconfigurable}.
An \ac{ris} comprises a large number of passive elements, with sub-wavelength inter-element spacing. Each element scatters the impinging signal after applying a phase shift to it. In order to improve the communication \ac{qos} in a wireless network comprising one or more \acp{ris}, the \acp{ris}' phase profiles should be jointly optimized \cite{abrardo2020intelligent}. Performing such optimization requires the \ac{csi} for the channel 
%each of the paths 
corresponding to each individual \ac{ris}. Furthermore, in localization, parameter estimation for each of the controlled channels as well as the uncontrolled one is often essential.  Obtaining the per-channel \ac{csi} entails resolving the interpath interference at the \ac{rx}, which is the topic of this paper. This task is  challenging  since \emph{i)}  unlike more conventional technologies such as relays, \acp{ris} are often passive and are not equipped with a \ac{rf} chain, thus they cannot transmit, receive, process, or amplify signals \cite{wu2021intelligent} \emph{ii)} \acp{ris} can only apply discrete phase shifts to the impinging signal \cite{pan2021reconfigurable}.

 Channel estimation for \ac{ris}-aided wireless systems has been studied in a number of papers (see \cite[Sec.\,VI-G]{di2020smart}). In order to estimate the channel from an individual \ac{ris}, it is generally assumed that the uncontrolled multipath is blocked (see e.g., \cite{alexandropoulos2020hardware})
 %cui2019efficient
 or that a \ac{ris} can be ``switched off''  (i.e., operate in absorption mode), thereby avoiding interpath interference (see e.g., \cite{mishra2019channel}).
 %nadeem2020intelligent,wang2020channel
 In \cite{jensen2020optimal,you2020channel}, the channel for each \ac{ris} element is estimated by designing the \ac{ris} phase profile according to orthogonal sequences. In all these studies the wireless system is aided by a single \ac{ris}. A limited number of papers consider channel estimation in multi-\ac{ris} systems: %(which is the focus of this work).% without assuming existing \ac{csi}. 
 in \cite{ning2020channel}, channel estimation is performed for a multi-\ac{ris} \ac{mimo} system via successive hierarchical beam sweeping for each of the \acp{ris}. In \cite{alexandropoulos2020phase} explicit multi-\ac{ris} channel estimation is bypassed by maximizing the achievable rate directly using a supervised learning technique.

Similarly as in  channel estimation, in multi-RIS  localization it has been considered switching off  \acp{ris} to avoid interpath interference \cite{wang2021joint}.  In \cite{keykhosravi2021semi}, a semi-passive localization scenario has been studied, where each of the users are equipped with an \ac{ris} and orthogonal \ac{ris} phase profiles are used to resolve the interpath interference. \acp{crb}  on the localization accuracy have been derived in \cite{elzanaty2020reconfigurable} for multi-\ac{ris} \ac{mimo} systems, however no estimation method is provided. In \cite{keykhosravi2020siso}  localization is performed via successive interference cancellation, assuming that the direct channel is much stronger than the channel reflected from an \ac{ris}.

In this paper, we tackle the problem of interpath interference during pilot transmission for a generic multi-\ac{ris} communication or localization scenario.
Our specific contributions are as follows: \textit{(i)} we design temporally coded \ac{ris} phase profiles with finite resolution, such that the signal received by different paths are mutually orthogonal, based on \ac{bh} matrices; \textit{(ii)} we show that these codes do not affect the chosen beamforming \ac{ris} phase profile, though they introduce a delay; \textit{(iii)} to address the delay, we study the problem of finding the shortest orthogonal code for finite resolution \ac{ris} using properties of \ac{bh} matrices; \textit{(iv)}
 to assess the effectiveness of our proposal, we consider the semi-passive localization scenario in \cite{keykhosravi2021semi} under limited \ac{ris} phase shifts resolution and show that  the \ac{crb}  bounds can be obtained with one bit phase resolution, while ignoring the  finite \ac{ris} phase resolution leads to  significant interpath interference.

\emph{Notation:} The vectors, which are columns by default, are represented by the lower-case bold letters and the matrices by the upper case bold ones. The $i$th element of vector $\bm{x}$ is shown by $[\bm{x}]_i$ and the element on the $i$th row and $j$th column of matrix $\mathbf{X}$ is shown by $[\mathbf{X}]_{i,j}$. Furthermore, $[\mathbf{X}]_{i}$ indicates the $i$th row of matrix $\mathbf{X}$. We show the the ceiling function by $\lceil\cdot\rceil$. We define the set $\mathbb{T}$ to represent the circle group and $\mathbb{T}_R$ to represent $\{e^{\jmath2\pi r/R}\}_{r=0}^{R-1}$. The  zero-mean  circularly  symmetric Gaussian  distribution with covariance  matrix $\bm{C}$ is represented by $\mathcal{CN}(\bm{0}, \bm{C})$. The Kronecker and Hadamard products are indicated by $\otimes$ and $\odot$ respectively. We define $\mathbf{X}^{\otimes n} = \underbrace{\mathbf{X} \otimes \dots \otimes \mathbf{X}}_{n \text{ times}}$ and $\mathbf{X}^{\otimes 0} =1$.

\section{System model}
We consider a general  \ac{mimo} scenario, where the \ac{tx} sends pilots with constant energy, for the purpose of channel estimation or user localization. The signal from the \ac{tx} is received  via the uncontrolled multipath and  controlled multipath from the $K\ge 1$ \acp{ris}. We ignore multi-bounce effects from two or more \acp{ris}. The received signal at $t$th transmission can be represented as
\begin{align}\label{eq:chMod1}
   \mathbf{y}_t= \Big(\mathbf{H}_0 + \sum\limits_{k=1}^K \mathbf{H}_k(\bm{\gamma}_{k,t})\Big)\sqrt{E_{\mathrm{s}}} +\bm{\nu}_{t}.
\end{align}
Here, $\mathbf{H}_0\in\mathbb{C}^{N_{\mathrm{rx}}\times N_{\mathrm{tx}}}$ represents the uncontrolled channel (which may or may not be blocked), where $N_{\mathrm{rx}}$ and $ N_{\mathrm{tx}}$ are the \ac{rx} and \ac{tx} antenna count. The channel from \ac{tx} to the $k$th \ac{ris} to the \ac{rx} is shown by $\mathbf{H}_k(\bm{\gamma}_{k,t})$, where $\bm{\gamma}_{k,t}\in \mathbb{T}^{N_\mathrm{ris}}_R$ represents the phase profile at time $t$ of  \ac{ris} $k$ with $N_\mathrm{ris}$ being the number of unit cells per \ac{ris} and $R \in \mathbb{N}_{>1}$ being the \acp{ris} phase resolution. The transmitted pilots are assumed to identical to $\sqrt{E_{\mathrm{s}}}$. Finally, $\bm{\nu}_{t}\sim \mathcal{CN}(\bm{0},N_0\bm{I}_{N_{\mathrm{rx}}})$ represents the complex  additive zero-mean white noise.
The channel related to the $k$th \ac{ris} can be decomposed as
\begin{align}
    \mathbf{H}_k(\bm{\gamma}_{k,t})&= \mathbf{H}_{\mathrm{sr},k} \mathrm{diag}(\bm{\gamma}_{k,t})\mathbf{H}_{\mathrm{ts},k},
\end{align}
where $\mathbf{H}_{\mathrm{sr},k} \in \mathbb{C}^{N_{\mathrm{rx}}\times N_{\mathrm{ris}}}$ and $\mathbf{H}_{\mathrm{ts},k} \in \mathbb{C}^{N_{\mathrm{ris}}\times N_{\mathrm{tx}}}$ represent the channels from \ac{ris} to \ac{rx} and from \ac{tx} to \ac{ris}, respectively. One can observe that for any $h\in\mathbb{T}_R$, we have
\begin{align}\label{nlos:h}
    \mathbf{H}_k(h\bm{\gamma}_{k,t}) = h\mathbf{H}_k(\bm{\gamma}_{k,t}).
\end{align}
Finally, we assume that the system designer has chosen desired phase profile $\bm{\zeta}_{k,q} \in \mathbb{T}^{N_{\mathrm{ris}}}_R$ for each \ac{ris} $k$, for transmissions times $q\in \{0,\ldots,Q-1\}$, $Q\ge 1$. The choice of $\bm{\zeta}_{k,q}$ and $Q$ is arbitrary, and could correspond, e.g., to $Q$ different beams that are directed to the \ac{rx} (e.g., if some prior knowledge of the \ac{rx} position is available); or to cover a plurality of users; or random beams to provide omnidirectional illumination. %The number of beams, $Q$ can be set based on the volume of the uncertainty region and the beam width.

Our proposed design (see Section \ref{sec:RIS_phase_profile_design}) is applicable to all the channel models that can be described by \eqref{eq:chMod1}--\eqref{nlos:h}, which includes most of the \ac{ris} channel models in the literature, whether narrow-band or wide-band, near-field or far-field, parametric or unstructured, single-RIS or multi-RIS, and with or without uncontrolled multipath.

\section{RIS phase profile design}\label{sec:RIS_phase_profile_design}

\subsection{Slow and fast \ac{ris} phase profiles}
We design the RIS phase profile by dividing the total transmission time $T$, into $P$ intervals with $Q$ symbols, i.e., we set $T = P\times Q$. Here, $P$ is a design parameters, which will be discussed later. We represent the \ac{ris} phase profile as
\begin{align}\label{eq:gama_kt}
    \bm{\gamma}_{k,t} = \beta_{k,p} \bm{\zeta}_{k,q},
\end{align}
where $p=t \, \mathrm{mod}\, P$ and $q = (t-p)/P$, so we have that $t=qP+p$. We call $\beta_{k,p}\in \mathbb{T}_R$ the \emph{fast-varying} part of the \ac{ris} phase profile and $\bm{\zeta}_{k,q} \in \mathbb{T}^{N_{\mathrm{ris}}}_R$ the \emph{slow-varying} part. Notice that since $\mathbb{T}_R$ is closed under multiplication, the $\bm{\gamma}_{k,t}$ in \eqref{eq:gama_kt} is in $\mathbb{T}^{N_{\mathrm{ris}}}_R$.

While the slow-varying parts are predetermined, we will design the fast varying part such that, for all $k=1\dots, K$, $k'\neq k$, we have 
%\begin{align}
 %   \sum\limits_{p=0}^{P-1} \beta_{k,p}\beta^*_{k',p} &=0\label{eq:beta1}\\
  %  \sum\limits_{p=0}^{P-1} \beta_{k,p} &=0\label{eq:beta2}\\
   % {\beta}_{k,p} & \in \mathbb{T}_R, \forall k,p \label{eq:beta3}.
%\end{align}
\begin{align}
    &\sum\limits_{p=0}^{P-1} \beta_{k,p}\beta^*_{k',p} =0\,\, \text{and}\,\, \sum\limits_{p=0}^{P-1} \beta_{k,p}=0 \label{eq:beta1}\\
    %\sum\limits_{p=0}^{P-1} \beta_{k,p} &=0\label{eq:beta2}\\
    &{\beta}_{k,p}  \in \mathbb{T}_R, \forall k,p \label{eq:beta3}.
\end{align}
With this design, in the \ac{rx} we first multiply the received signal with the conjugate of the transmitted pilot to obtain $\mathbf{w}_t = \mathbf{y}_t $. Then we calculate
\begin{align}
    \mathbf{z}_{k,q} &= \sum\limits_{p=0}^{P-1}\beta^*_{k,p} \mathbf{w}_{qP+p} =P \sqrt{E_{\mathrm{s}}} \mathbf{H}_k(\bm{\zeta}_{k,q})\mathbf{1}_{N_{\mathrm{tx}}} + \bm{\nu}_{k,q}\label{eq:zkq}\\
    \mathbf{z}_{0,q} &= \sum\limits_{p=0}^{P-1} \mathbf{w}_{qP+p} =P \sqrt{E_{\mathrm{s}}} \mathbf{H}_0\mathbf{1}_{N_{\mathrm{tx}}} + \bm{\nu}_{0,q},\label{eq:zkq}
\end{align}
where, $\mathbf{1}_{N_{\mathrm{tx}}}$ denotes the all-one vector of length $N_{\mathrm{tx}}$. Vectors $\bm{\nu}_{k,q}$ and $\bm{\nu}_{0,q}$ \Corr{represent the additive noise with distribution $\mathcal{CN}(\bm{0},PN_0\bm{I}_{N_{\mathrm{rx}}})$}. It can be seen that  in the observations $\mathbf{z}_{k,q}$ all interpath interference is eliminated and that each path can be estimated separately. The question remains how to design the fast-varying part to satisfy \eqref{eq:beta1}--\eqref{eq:beta3}, while simultaneously keeping $P$ as small as possible. This will be discussed next.

\subsection{Designing the fast-varying part}\label{sec:fast_varying}

In this section, we show how to find a feasible  solution of $\beta_{k,p}\in \mathbb{T}_{R}$ that satisfy \eqref{eq:beta1}--\eqref{eq:beta3} for a given $K$ and $R$. Also, we consider the problem of finding $P^*$, the minimum value of $P$ that allows a feasible solution. We refer to such solution as the \emph{optimal solution}.
We will now define several concepts, based on which we show for which cases an optimal solution can be found, while providing a simple method for finding a feasible solution for the general case.

We first introduce $\mathrm{BH}(P,R)$ as the set of \ac{bh} matrices of order $P$ and complexity $R$, that is all matrices  $\bm{\Upsilon}\in \mathbb{T}^{P\times P}_R$ that satisfy $\bm{\Upsilon} \bm{\Upsilon}^{\mathrm{H}} = P\mathbf{I}_P$, where $\mathbf{I}_P$ is a $P\times P$ identity matrix. In general, these sets may be empty and finding elements in $\mathrm{BH}(P,R)$ or characterizing the entire set is an open problem. However, in \cite{lampio2020orderly}, a numerical method is introduced to search for BH matrices; the authors publish a list of discovered BH matrices on their Wiki page\footnote{\url{https://wiki.aalto.fi/display/Butson}}. 

Secondly, we represent the prime factorization of $R$ by%\footnote{In this paper, we let $R$ be any integer larger than 1. Therefore, we do not make any assumption on the RIS control technology that exists now or may come into existence in future.} 
\begin{align}\label{eq:primeFac}
    R = p_1^{e_1}\times p_2^{e_2}\times \dots \times p_L^{e_L},
\end{align} 
where we assume that $p_1< p_2 < \dots < p_L$. 

Thirdly, we define the vectors $\bm{\beta}_k = [\beta_{k,0}, \beta_{k,1}, \dots, \beta_{k,P-1}]^\top$ and $\bm{\beta}_0 = \mathbf{1}_P$. Then the conditions \eqref{eq:beta1} %and \eqref{eq:beta2} 
can be rewritten as 
\begin{align}
\bm{\beta}_0 &= \mathbf{1}_P\label{eq:bVecCondition1}\\
    \bm{\beta}_k^{\mathrm{H}} \bm{\beta}_{k'} &= P \mathbbm{1}(k=k'),\label{eq:bVecCondition}
\end{align}
where $\mathbbm{1}(\cdot)$ is the indicator function. Furthermore, we let $\mathbf{B}\in \mathbb{T}_{R}^{(K+1)\times P} $ be $[\bm{\beta}_0^\top, \dots, \bm{\beta}_{K}^{\top}]^{\top}$ and rewrite the conditions \eqref{eq:beta1} % and \eqref{eq:beta2} 
as
\begin{align}
    [\mathbf{B}]_1 &= \mathbf{1}_P^\top\label{eq:bMatCondition1}\\
    \mathbf{B}\mathbf{B}^{\mathrm{H}} &= P \mathbf{I}_{K+1}.\label{eq:bMatCondition2}
\end{align}

\subsection{Special cases}
Before dealing with the general case, we consider five  specific cases of practical importance for which an optimal solution to \eqref{eq:beta1}--\eqref{eq:beta3} can be found.
\subsubsection{Case \emph{I}: Infinite resolution: $\beta_{k,p}\in \mathbb{T}$}\label{sec:infb}
With infinite phase resolution, we can find a feasible solution by setting $\bm{\beta}_k^\top$ to be the $k$th row of the $P\times P$ \ac{dft} matrix $\mathbf{F}_P$, i.e.,
\begin{align}\label{eq:binf}
    \beta_{k,p} = [\mathbf{F}_P]_{k,p} = e^{-\jmath 2\pi kd /P}.
\end{align}
The condition \eqref{eq:bVecCondition} holds since the \ac{dft} matrix is unitary.   This also shows, by construction, that $P^*\le K+1$. Via the following lemma, we prove that the minimum duration solution is $P^*=K+1$. 

\begin{lemma}\label{lemma:kPlus1}
If $P<K+1$, then there is no solution  for $\bm{\beta}_k \in \mathbb{T}_R^{P}$ that satisfies \eqref{eq:bVecCondition} for $k=0, \dots, K$.
\end{lemma}
\begin{proof}
 Since $\bm{\beta}_k\in \mathbb{C}^{P}$ then, if $K+1>P$, vectors $\bm{\beta}_k$ must be linearly dependent. This is in contradiction to \eqref{eq:bVecCondition}.
\end{proof}

\subsubsection{Case \emph{II}: $p_1\geq K+1$}\label{sec:caseII}
We now show that in this case, $P^* = p_1$. To do so, we first present a solution for $\bm{\beta}_k$ with $P=p_1$ and then prove that no solutions exist for $P<p_1$. The first part is done by setting  $\bm{\beta}_k^\top$ to the $k$th row of $\mathbf{F}_{p_1}$. Since $\mathbb{T}_{p_i} \subseteq \mathbb{T}_{R}$, this yields a valid solution.
To prove the second part, we use the following proposition.

\begin{proposition}\label{prob:1}
For any $R>1$ with the prime factorization \eqref{eq:primeFac}, let $\mathcal{W}(R)$ be the set of all positive integers $M$ for which there exists $\bm{\alpha} \in \mathbb{T}_R^M$ such that $\bm{\alpha}^\top \mathbf{1}_M=0$. Then 
%\begin{align}\label{eq:Wdef2}
%\exists m_i \in \mathbb{N}:M = \sum_{i=1}^Lm_ip_i \neq 0
%\end{align}
\begin{align}\label{eq:Wdef}
    \mathcal{W}(R) = \Big\{\sum_{i=1}^Lm_ip_i\, \vert\, m_i\in\mathbb{N} \Big\} \setminus \{0\}.
\end{align}
\end{proposition}
\begin{proof}
See  \cite{lam2000vanishing}.
\end{proof}
Now, assume that $P < p_1$ and that $\bm{\beta}_k \in \mathbb{T}^P_R$ satisfies \eqref{eq:bVecCondition1}--\eqref{eq:bVecCondition}. Then, $\bm{\beta}_1^\top \mathbf{1}_P =0$, which,  according to Proposition~\ref{prob:1}, implies that $P\in \mathcal{W}(R)$. However, from \eqref{eq:Wdef}, the minimum member of $\mathcal{W}(R)$ is $p_1$ and this is in contradiction with the assumption that $P < p_1$.

\subsubsection{Case \emph{III}: $R=2$}\label{sec:case:R2}
In this case, $\bm{\beta}_k$ are real vectors, and therefore we should consider (real) Hadamard matrices in order to find a solution. If $K=1$, one can check that $\mathbf{F}_2$ is an optimal solution. For higher values of $K$, we use the following proposition. 
\begin{proposition}\label{prop:2}
If $\mathrm{BH} (P,2)$ is non-empty, and $P>2$ then a positive integer $\kappa$ exists such that $P=4\kappa$.
\end{proposition}
\begin{proof}
See \cite[Sec.\,1.3.2]{laclair2016survey}.
\end{proof}
The inverse of Proposition~\ref{prop:2} reads as follows. 
\begin{conjecture}[Hadamard Conjecture]\label{conj:1}
For any positive integer $\kappa$, the set $\mathrm{BH} (4\kappa,2)$ is non-empty.
\end{conjecture}
This conjecture has been an open problem for over $100$ years. However, it has been proven for values $\kappa<167$ \cite{laclair2016survey}, which is enough for most practical purposes. Next, we present the following theorem on optimal solutions for $R=2$.
\begin{theorem}
Let $2<K<664$. Then the smallest value of $P$ that yields a solution for $\{\bm{\beta}_k\in\mathbb{T}_{2}^P\}_{k=0}^{K}$ that satisfies \eqref{eq:bVecCondition1} and \eqref{eq:bVecCondition} is  $ P^* = 4\lceil(K+1)/4\rceil$. 
\end{theorem}
\begin{proof}
See Appendix \ref{app:Thm1}.
\end{proof}
To generate the optimal solution one should assign $\bm{\beta}_k$ (for $k=0, \dots, K$) to the $k$th row of a matrix  $\mathrm{BH}(P^*, 2)$ that satisfies \eqref{eq:bMatCondition1}.

\subsubsection{Case \emph{IV}: $R=2^r$ with $r>1$}
We begin by stating a conjecture about the existence of $\mathrm{BH}(P,4)$ matrices \cite{seberry1973complex}.
\begin{conjecture}\label{conj:turyn}
For all positive integers $\kappa$, the set $\mathrm{BH}(2\kappa,4)$ is non-empty.
\end{conjecture}
It was shown that this conjecture is true for all $\kappa\leq32$ \cite{seberry1973complex}, which is large enough for all the practical cases that we consider in this paper. The following corollary follows from the fact that for all $r\geq 2$ we have that $\mathrm{BH}(P,4)\subseteq\mathrm{BH}(P,2^r)$.
\begin{corollary}\label{cor:two-to-r}
For all positive integers $\kappa\leq 32$ and  $r\geq 2$, the set $\mathrm{BH}(2\kappa,2^r)$ is non-empty. 
\end{corollary}
Next, in the following theorem we discuss the optimal solution.
\begin{theorem}
Let $1\leq K<64$ and $r\geq 2$, then the smallest value of $P$, for which there exists a solution for $\{\bm{\beta}_k\in\mathbb{T}_{2^r}^P\}_{k=0}^{K}$ that satisfies \eqref{eq:bVecCondition1} and \eqref{eq:bVecCondition} is  $ P^*=2\lceil(K+1)/2\rceil$. 
\end{theorem}
\begin{proof}
See Appendix \ref{app:Thm2} 
\end{proof}
The optimal solution can be generated based on the rows of a $\mathrm{BH}(P^*, 2^r)$ matrix satisfying \eqref{eq:bMatCondition1}.

\subsubsection{Case \emph{V}: $\mathrm{BH}(K+1,R)$ is non-empty}
We first introduce the following lemma from \cite{butson1962generalized} for later use.
\begin{lemma}\label{lemma:1}
If there exists $\mathbf{G}\in \mathbb{T}^{(K+1)\times P}_R$ that satisfies  \eqref{eq:bMatCondition2}, then there exists a matrix $\mathbf{B}\in \mathbb{T}^{(K+1)\times P}_R$ that satisfies both \eqref{eq:bMatCondition1} and \eqref{eq:bMatCondition2} and it can be generated as $ \mathbf{G} \bm{\Lambda}^{\mathrm{H}}$,
where $\bm{\Lambda} = \text{diag}([\mathbf{G}]_1)$.
\end{lemma}
\begin{proof}
The proof follows by simply substituting $\mathbf{B}$ into \eqref{eq:bMatCondition1} and \eqref{eq:bMatCondition2}.
\end{proof}
 Let now $\mathbf{A}\in\mathrm{BH}(K+1,R)$. Then by using Lemma~\ref{lemma:1} we can generate matrix $\mathbf{B}$ whose rows represent a solution for $\bm{\beta}_k$ in $\mathbb{T}_B^{K+1}$. Therefore, we have that $P^* \leq K+1$. However, based on Lemma~\ref{lemma:kPlus1}, we have  $P^* \geq K+1$, and thus $P^* = K+1$.

\subsection{General case}
We  note that based on Lemma~\ref{lemma:1}, we can relax the condition \eqref{eq:bMatCondition1} and only search for matrices $\mathbf{B} \in \mathbb{T}^{(K+1)\times P}_R$ that satisfies \eqref{eq:bMatCondition2}. We call the set of such matrices partial BH  (PBH) matrices of order $(K+1,P)$ and complexity $R$ and denote it by $\mathrm{PBH}(K+1,P,R)$. It is evident that $\mathrm{PBH}(P,P,R) =\mathrm{BH}(P,R)$. Thus, since finding BH matrices, in general, is an open problem, so is finding PBH ones. 
One way to generate a $\mathrm{PBH}(K+1,P,R)$ is by selecting the first $K+1$ rows of a $\mathrm{BH}(P,R)$ matrix with $P\geq K+1$. In what follows we present a simple method to generate this BH matrix using the \ac{dft} matrices.

We define the set 
\begin{align}
    \mathcal{V}(R) = \Big\{\prod_{i=1}^Lp_i^{m_i}\geq K+1 \big\vert\,  m_i\in\mathbb{N} \Big\}
\end{align}
and determine $\tilde{v} = \min\{\mathcal{V}(R)\}  = \prod_{i=1}^Lp_i^{\tilde{m}_i}$.
%\begin{align}
%    \tilde{v} = \min\{\mathcal{V}(B)\}  = \prod_{i=1}^Lp_i^{\tilde{m}_i}
%\end{align}
Then 
$
    \mathbf{S} = \mathbf{F}_{p_1}^{\otimes \tilde{m}_1} \otimes \dots \otimes \mathbf{F}_{p_L}^{\otimes \tilde{m}_L}
$
is a $\mathrm{BH}(\tilde{v},R)$ matrix that satisfies \eqref{eq:bMatCondition1}. Therefore, one can set $\bm{\beta}_k^{\top}$ to the first $K+1$ rows of $\mathbf{S}$. In general, this solution may not be optimal. However, if $p_1\geq K+1$ then $\mathbf{S} = \mathbf{F}_{p_1}$, which, as proved in Sec.\,\ref{sec:caseII}, is the optimal solution.

\subsection{Summary}
We have shown for 5 special cases that optimal solutions can be found and provided a method to find such solution. For the general case, we can only provide a constructive method for finding feasible solutions, though without guarantee of optimality. 
\section{Simulation results}

\begin{figure}
    \centering
    \begin{tikzpicture}
    \node  [anchor=mid] (pic) {\includegraphics[width=5cm]{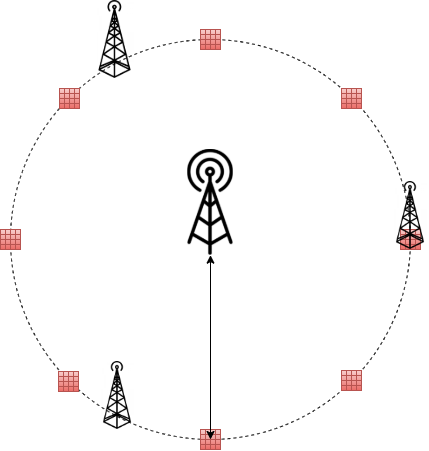}};
    \node (dist) at (.4, 1.4) {\footnotesize $10$m};
    \node (UE) at (2, 4) {\footnotesize RIS};
    \node (Rx) at (-1.5, 4.8) {\footnotesize Rx};
    \node (Bs) at (.4, 3) {\footnotesize BS};
    \end{tikzpicture}
    \caption{System setup with $8$ \acp{ris} and three Rxs located on a circle of radius $10$ m around the BS. The \acp{ris}  and Rxs are on $z=-3$~m and  $z=1$~m planes, respectively and BS is at origin.}
    \label{fig:system}
\end{figure}
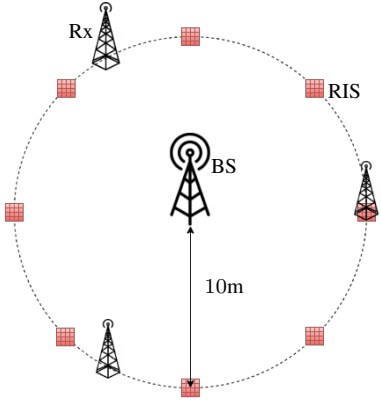

\begin{figure}
    \centering
    \input{cdfs}
    \caption{ Cumulative distribution of root mean squared error ($e$ in meter) of the estimator over $100$ random realizations of RIS phase profiles. The results are presented for our method with  one-bit RIS phase resolution ($R=2$) and for the  method in \cite{keykhosravi2021semi}, followed by quantization with $R\in\{1,4,\dots,64\}$.  The theoretical lower bound is shown by the dashed line. }
    \label{fig:PebResults}
\end{figure}
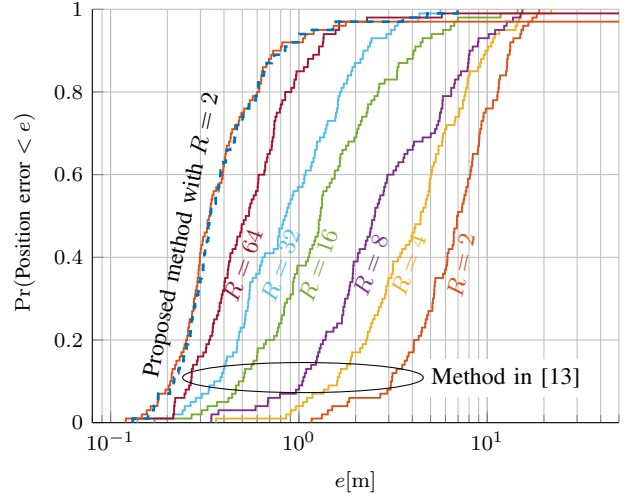

In this section, we illustrate the effectiveness of the proposed method in mitigating the interpath interference in a semi-passive localization scenario from \cite{keykhosravi2021semi}. 
\subsection{Scenario}
In this setup we consider $K=8$  \acp{ris}, one \ac{bs} transmitting pilot signals, and $3$ receivers that receive the signal directly from the \ac{bs} and also the reflected signal from the \acp{ris}. The placement of each of the elements is depicted in Fig.\,\ref{fig:system}. We consider the transmission of \ac{ofdm} signals from the \ac{bs}. The parameters of the signal and noise can be found in \cite[Table\,I]{keykhosravi2021semi}. 
We set $Q=1$ and select $\bm{\zeta}_{k,0}$ randomly. 
We then estimate the location of each of the  $K$ \acp{ris} as follows. We first resolve the interpath interference at each \ac{rx} via  orthogonal coding of the \ac{ris} phase profiles as described in Section\,\ref{sec:RIS_phase_profile_design}. Then, we calculate the \ac{toa} of all the $K+1$ paths ariving at each of the \acp{rx} (see \cite[Sec.\,III-B]{keykhosravi2021semi}). Finally, we estimate the RIS location based on the calculated \acp{toa} (see \cite[III-C]{keykhosravi2021semi}).
%
%Since in this scenario we only estimate \acp{toa} and estimating delays do not require varying \ac{ris} phase profiles (as oppose to estimating \acp{aod}), we set the number of the slow-varying part of the \ac{ris} phase profile (i.e, $\bm{\zeta}_{k,0}$) equal to one (i.e., $Q=1$). We assume that no prior information of the user position is available and therefore $\bm{\zeta}_{k,0}$ is selected randomly.
\subsection{Fast-varying phase profiles}
With regard to the fast-varying part of the \ac{ris} phase profile (i.e, $\bm{\beta_k}$), we consider two different methods:
\begin{itemize}
    \item \emph{Proposed method:} The first method is based on our proposal in Section\,\ref{sec:fast_varying}, where we assume that \acp{ris} have one bit of phase resolution (i.e., $R=2$). Based on our results in Section\,\ref{sec:case:R2} we set $P=12$ and we let $\bm{\beta_k}$ be the rows of a $\mathrm{BH}(2,12)$ matrix (or equivalently, a $12\times 12$ Hadamard matrix).
    \item \emph{Benchmark:} With the second method, we ignore the limited phase resolution and  design  $\bm{\beta_k}$ as in Sec.\,\ref{sec:infb} and  \cite{keykhosravi2021semi}. For a fair comparison, we consider the same number of transmission as in the proposed method, i.e., $P=12$. Therefore, we set $\bm{\beta_k}$ to be the rows of the matrix $\mathbf{F}_{12}$. We consider six different phase resolutions $R\in\{2,4,8,16,32,64\}$. For each $R$, we map the values of $\beta_{k,p}$ (obtained from $\mathbf{F}_{12}$) to the closest point in $\mathbb{T}_R$ to account for limited phase resolution of the \acp{ris}. Such quantization impairs the orthogonality of $\bm{\beta_k}$. Therefore, we expect that the estimator cannot avoid interpath interference and its accuracy deteriorates. 
\end{itemize}

\subsection{Results and discussion}

In Fig.\,\ref{fig:PebResults}, we demonstrate the \ac{cdf} of the \ac{ris} position estimation error for $100$ random realizations of $\bm{\zeta}_{k,0}  \in \mathbb{T}^{N_{\mathrm{ris}}}_R$. We present the results for the \ac{ris} located at $(10, 0, -3)$.
As it can be seen, the performance degradation due to quantization is more severe with low values of $R$. On the other hand, our method takes into account the limited resolution of the \ac{ris} phase shifts and therefore, avoids any quantization effects.  It can be seen that via our proposed method, with only one-bit resolution, the estimator can attain the theoretical lower bound, which is calculated through \ac{fim} analysis \cite[Chapter~3]{kay1993fundamentals}.

\section{Conclusion}
We proposed a method to resolve interpath interference in wireless systems equipped with multiple \acp{ris} with limited phase resolution. We did so by dividing the \ac{ris} phase profiles into a fast-varying and a slow-varying part and designing the former part to be orthogonal across \acp{ris}. We proposed a method to find such orthogonal sequences in general and discussed the special cases where such sequences are optimal (have the shortest length). Our solution can be applied to 
most system models considered in the literature. 
%large variety of system models such as near(far)-field, wide(narrow)-band, and multiple(single)-input multiple(single)-output systems. 
Future work can explore a trade-off between the interpath interference and number of transmissions ($P$) and look for sequences ($\bm{\beta_k} \in \mathbb{T}_R^P$) with minimum correlation for a given $P$.

\appendices
\section{Proof of Theorem 1} \label{app:Thm1}
Let $U = 4\lceil(K+1)/4\rceil$ and $P^*$ be the smallest value of $P$ that yields a solution.
From Conjecture\,\ref{conj:1} and Lemma~\ref{lemma:1}, we know that matrix $\mathbf{B}$ exists such that $\mathbf{B}\in \mathrm{BH}(U, 2)$ and $[\mathbf{B}]_1$ satisfies \eqref{eq:bVecCondition1}. Therefore, the first $K+1$ rows of $\mathbf{B}$ provide a solution for $\bm{\beta}_k$ (note that the dimension of $\mathbf{B}$ is at least $K+1$). Furthermore, from Lemma\,\ref{lemma:kPlus1} we know that $P^*\geq K+1$. 
To complete the proof, we need to  show that any $P$ that yields a solution for $\bm{\beta}_k$ must be a factor of $4$. The proof is in the same lines as the proof of Proposition~\ref{prop:2} in \cite[Sec.\,1.3.2]{laclair2016survey} and we present it here for completeness. 

 Based on \eqref{eq:bVecCondition1} and \eqref{eq:bVecCondition}, we have that $\bm{\beta}_2^\top\mathbf{1}_P=0$, therefore, $P$ should be even and  $\bm{\beta}_2$ contain equal number of $1$ and $-1$. Without loss of generality we can assume that  $\bm{\beta}_2=[\mathbf{1}_{P/2}^\top,-\mathbf{1}_{P/2}^\top]$. Then since $\bm{\beta}_3^\top\bm{\beta}_2=0$ we have that $\Sigma_1-\Sigma_2=0$, where $\Sigma_1$ ($\Sigma_2$) is the summation of all the elements in the first (second) half of vector $\bm{\beta}_3$. Furthermore, since $\bm{\beta}_3^\top\bm{\beta}_1=0$, we get $\Sigma_1+\Sigma_2=0$. Therefore, we have $\Sigma_1=\Sigma_2=0$ which proves that each half of the vector $\bm{\beta}_3$ must contain equal number of $1$ and $-1$. Therefore, the length of $\bm{\beta}_3$, which is $P$, must be divisible by $4$.

\section{Proof of Theorem 2} \label{app:Thm2}
Let $U=2\lceil(K+1)/2\rceil$ and $P^*$ be the smallest value of $P$ that yields a solution. Based on Corollary\,\ref{cor:two-to-r} and Lemma\,\ref{lemma:1}, there exists $\mathbf{B}\in \mathrm{BH}(U,2^r)$ whose first $K+1$ rows determine a solution for $\{\bm{\beta}_k\in\mathbb{T}_{2^r}^P\}_{k=0}^{K}$ (note that $U\geq K+1$). Therefore, $P^*\leq U$. 
Furthermore, from \eqref{eq:bVecCondition1} and  \eqref{eq:bVecCondition}, we have that  $\bm{\beta}_1^\top\mathbf{1}_{P^*} = 0$. Therefore, from Proposition\,\ref{prob:1} we have that $P^*\in \mathcal{W}(2^r)$, which is the set of positive even numbers (see \eqref{eq:Wdef}). Also, from Lemma\,\ref{lemma:kPlus1} we know that $P^*\geq K+1$. Therefore, $P^*$ must be larger than or equal to the smallest even number lager than $K+1$, that is $U$.
%==============================================================
\begin{comment}
\begin{align}
    \begin{bmatrix}
     1& 1& 1& 1& 1& 1& 1& 1& 1& 1& 1& 1\\
     1&-1& 1&-1& 1& 1& 1&-1&-1&-1& 1&-1\\
     1&-1&-1& 1&-1& 1& 1& 1&-1&-1&-1& 1\\
     1& 1&-1&-1& 1&-1& 1& 1& 1&-1&-1&-1\\
     1&-1& 1&-1&-1& 1&-1& 1& 1& 1&-1&-1\\
     1&-1&-1& 1&-1&-1& 1&-1& 1& 1& 1&-1\\
     1&-1&-1&-1& 1&-1&-1& 1&-1& 1& 1& 1\\
     1& 1&-1&-1&-1& 1&-1&-1& 1&-1& 1& 1\\
     1& 1& 1&-1&-1&-1& 1&-1&-1& 1&-1& 1\\
     1& 1& 1& 1&-1&-1&-1& 1&-1&-1& 1&-1\\
     1&-1& 1& 1& 1&-1&-1&-1& 1&-1&-1& 1\\
     1& 1&-1& 1& 1& 1&-1&-1&-1& 1&-1&-1
\end{bmatrix} 
\end{align}
%==============================================================
\end{comment}

\vspace{-.5cm}

\bibliographystyle{IEEEtran}
\bibliography{references}

\end{document}

%% file: acros.tex
\acrodef{cdf}[CDF]{cumulative distribution function}
\acrodef{bs}[BS]{base station}
\acrodefplural{bs}[BSs]{base station}
\acrodef{ue}[UE]{user equipment}
\acrodef{los}[LOS]{line-of-sight}
\acrodef{aoa}[AoA]{angle-of-arrival}
\acrodefplural{aoa}[AoA]{angles-of-arrival}
\acrodef{aod}[AoD]{angle-of-departure}
\acrodef{toa}[ToA]{time-of-arrival}
\acrodef{tdoa}[TDoA]{time-difference-of-arrival}
\acrodef{ris}[RIS]{reconfigurable intelligent surface}
\acrodefplural{ris}[RISs]{reconfigurable intelligent surfaces}
\acrodef{tx}[Tx]{transmitter}
\acrodef{rx}[Rx]{receiver}
\acrodefplural{rx}[Rxs]{receivers}
\acrodef{crb}[CRB]{Cram\'er-Rao lower bound}
\acrodef{rss}[RSS]{received signal strength}
\acrodef{los}[LOS]{line-of-sight}
\acrodef{nlos}[NLOS]{non line-of-sight}
\acrodef{dft}[DFT]{discrete Fourier transform}
\acrodef{fft}[FFT]{fast Fourier transform}
\acrodef{fim}[FIM]{Fisher information matrix}
\acrodef{upa}[UPA]{uniform planar array}
\acrodef{peb}[PEB]{position error bound}
\acrodef{snr}[SNR]{signal-to-noise ratio}
\acrodef{sre}[SRE]{smart radio environment}
\acrodefplural{sre}[SRE]{smart radio environments}
\acrodef{mimo}[MIMO]{multiple-input  multiple-output}
\acrodef{miso}[MISO]{multiple-input  single-output}
\acrodef{rf}[RF]{radio-frequency}
\acrodef{qos}[QoS]{quality-of-service}
\acrodef{rfid}[RFID]{radio-frequency identification}
\acrodef{ofdm}[OFDM]{orthogonal frequency-division multiplexing}
\acrodef{pin}[PIN]{positive-intrinsic negative}
\acrodef{bh}[BH]{Butson-type Hadamard}
\acrodef{csi}[CSI]{channel state information}

%% file: cdfs.tex
% This file was created by matlab2tikz.
%
%The latest updates can be retrieved from
%  http://www.mathworks.com/matlabcentral/fileexchange/22022-matlab2tikz-matlab2tikz
%where you can also make suggestions and rate matlab2tikz.
%
\definecolor{mycolor1}{rgb}{0.00000,0.44700,0.74100}%
\definecolor{mycolor2}{rgb}{0.85000,0.32500,0.09800}%
\definecolor{mycolor3}{rgb}{0.92900,0.69400,0.12500}%
\definecolor{mycolor4}{rgb}{0.49400,0.18400,0.55600}%
\definecolor{mycolor5}{rgb}{0.46600,0.67400,0.18800}%
\definecolor{mycolor6}{rgb}{0.30100,0.74500,0.93300}%
\definecolor{mycolor7}{rgb}{0.63500,0.07800,0.18400}%

\pgfplotsset{every tick label/.append style={font=\footnotesize}}

\begin{tikzpicture}
\begin{axis}[%
width=7cm,
height=5.5cm,
at={(9cm,-6.1cm)},
scale only axis,
unbounded coords=jump,
xmode=log,
xmin=0.08,
xmax=50,
xminorticks=true,
xlabel={\footnotesize $e [\mathrm{m}]$},
ymin=0,
ymax=1,
axis x line*=bottom,
axis y line*=left,
xmajorgrids,
xminorgrids,
ymajorgrids,
axis background/.style={fill=white},
legend style={at={(0.5,0.05)}, anchor=south west, legend cell align=left, align=left, draw=white!15!black, font=\footnotesize}
]

\addplot [color=mycolor2,line width=0.7pt]
  table[row sep=crcr]{%
-inf	0\\
1.16787365686267	0\\
1.16787365686267	0.01\\
1.32474230928671	0.01\\
1.32474230928671	0.02\\
1.42410544164953	0.02\\
1.42410544164953	0.03\\
1.5027736056903	0.03\\
1.5027736056903	0.04\\
1.81089411512606	0.04\\
1.81089411512606	0.05\\
1.84954909633516	0.05\\
1.84954909633516	0.06\\
2.70242442275128	0.06\\
2.70242442275128	0.07\\
2.94694472296698	0.07\\
2.94694472296698	0.08\\
3.0690487488163	0.08\\
3.0690487488163	0.09\\
3.07684552769163	0.09\\
3.07684552769163	0.1\\
3.12407131884373	0.1\\
3.12407131884373	0.11\\
3.13918438030488	0.11\\
3.13918438030488	0.12\\
3.29591521618505	0.12\\
3.29591521618505	0.13\\
3.53573020358938	0.13\\
3.53573020358938	0.14\\
3.64557948953	0.14\\
3.64557948953	0.15\\
3.75870119211051	0.15\\
3.75870119211051	0.16\\
3.82386399367971	0.16\\
3.82386399367971	0.17\\
4.06453899098465	0.17\\
4.06453899098465	0.18\\
4.07452945288692	0.18\\
4.07452945288692	0.19\\
4.08389867363995	0.19\\
4.08389867363995	0.2\\
4.32016459841649	0.2\\
4.32016459841649	0.21\\
4.46039340983517	0.21\\
4.46039340983517	0.22\\
4.50284834597488	0.22\\
4.50284834597488	0.23\\
4.68244120079959	0.23\\
4.68244120079959	0.24\\
4.71278101577034	0.24\\
4.71278101577034	0.25\\
4.79181979267582	0.25\\
4.79181979267582	0.26\\
4.86709026961029	0.26\\
4.86709026961029	0.27\\
5.10965868944114	0.27\\
5.10965868944114	0.28\\
5.11699654912678	0.28\\
5.11699654912678	0.29\\
5.14822024385771	0.29\\
5.14822024385771	0.3\\
5.21059953430936	0.3\\
5.21059953430936	0.31\\
5.23586515932049	0.31\\
5.23586515932049	0.32\\
5.23831280734574	0.32\\
5.23831280734574	0.33\\
5.2642329199028	0.33\\
5.2642329199028	0.34\\
5.30119124020178	0.34\\
5.30119124020178	0.35\\
5.81802104152454	0.35\\
5.81802104152454	0.36\\
5.96968647280478	0.36\\
5.96968647280478	0.37\\
6.05604842444998	0.37\\
6.05604842444998	0.38\\
6.12695288947849	0.38\\
6.12695288947849	0.39\\
6.14582054126901	0.39\\
6.14582054126901	0.4\\
6.1494566082076	0.4\\
6.1494566082076	0.41\\
6.23217661795605	0.41\\
6.23217661795605	0.42\\
6.53890164991907	0.42\\
6.53890164991907	0.43\\
6.62598392458806	0.43\\
6.62598392458806	0.44\\
6.63481687106712	0.44\\
6.63481687106712	0.45\\
6.63836037251616	0.45\\
6.63836037251616	0.46\\
6.68554709050542	0.46\\
6.68554709050542	0.47\\
6.89590290340052	0.47\\
6.89590290340052	0.48\\
6.90343514458125	0.48\\
6.90343514458125	0.49\\
6.9105181073199	0.49\\
6.9105181073199	0.5\\
6.99598376306207	0.5\\
6.99598376306207	0.51\\
7.35278014394732	0.51\\
7.35278014394732	0.52\\
7.3745614119482	0.52\\
7.3745614119482	0.53\\
7.49264000246401	0.53\\
7.49264000246401	0.54\\
7.55940791090926	0.54\\
7.55940791090926	0.55\\
7.67969020071904	0.55\\
7.67969020071904	0.56\\
7.74990744477199	0.56\\
7.74990744477199	0.57\\
7.89006997080629	0.57\\
7.89006997080629	0.58\\
7.89667623727504	0.58\\
7.89667623727504	0.59\\
8.09265678970478	0.59\\
8.09265678970478	0.6\\
8.15164573200901	0.6\\
8.15164573200901	0.61\\
8.22339578652222	0.61\\
8.22339578652222	0.62\\
8.30296374871702	0.62\\
8.30296374871702	0.63\\
8.34864501275125	0.63\\
8.34864501275125	0.64\\
8.38189015808715	0.64\\
8.38189015808715	0.65\\
8.91614768445628	0.65\\
8.91614768445628	0.66\\
8.95877851208136	0.66\\
8.95877851208136	0.67\\
9.04113412230207	0.67\\
9.04113412230207	0.68\\
9.06924660627771	0.68\\
9.06924660627771	0.69\\
9.09726237481157	0.69\\
9.09726237481157	0.7\\
9.17142571549078	0.7\\
9.17142571549078	0.71\\
9.2987433125118	0.71\\
9.2987433125118	0.72\\
9.32442650546699	0.72\\
9.32442650546699	0.73\\
9.53109000142704	0.73\\
9.53109000142704	0.74\\
9.57174046166751	0.74\\
9.57174046166751	0.75\\
9.81380794059977	0.75\\
9.81380794059977	0.76\\
10.7074882339983	0.76\\
10.7074882339983	0.77\\
10.7293854660141	0.77\\
10.7293854660141	0.78\\
11.3002246337893	0.78\\
11.3002246337893	0.79\\
11.3437682893255	0.79\\
11.3437682893255	0.8\\
11.4642776007361	0.8\\
11.4642776007361	0.81\\
11.522446258791	0.81\\
11.522446258791	0.82\\
11.5584469475383	0.82\\
11.5584469475383	0.83\\
11.6820829361964	0.83\\
11.6820829361964	0.84\\
12.5806148392267	0.84\\
12.5806148392267	0.85\\
12.6886234032624	0.85\\
12.6886234032624	0.86\\
12.7540494471223	0.86\\
12.7540494471223	0.87\\
12.7594613514244	0.87\\
12.7594613514244	0.88\\
12.8576061100327	0.88\\
12.8576061100327	0.89\\
13.1482636456957	0.89\\
13.1482636456957	0.9\\
13.2470968466385	0.9\\
13.2470968466385	0.91\\
13.2575339084261	0.91\\
13.2575339084261	0.92\\
13.7693477074374	0.92\\
13.7693477074374	0.93\\
14.9730832899806	0.93\\
14.9730832899806	0.94\\
15.1718004755114	0.94\\
15.1718004755114	0.95\\
15.5110070362365	0.95\\
15.5110070362365	0.96\\
16.5296551033748	0.96\\
16.5296551033748	0.97\\
17.3585527293449	0.97\\
17.3585527293449	0.98\\
18.3244000999524	0.98\\
18.3244000999524	0.99\\
18.8816870009278	0.99\\
18.8816870009278	1\\
inf	1\\
};
%\addlegendentry{2}

\addplot [color=mycolor3,line width=0.7pt]
  table[row sep=crcr]{%
-inf	0\\
0.363261081246882	0\\
0.363261081246882	0.01\\
0.856303110498625	0.01\\
0.856303110498625	0.02\\
0.906070815633235	0.02\\
0.906070815633235	0.03\\
0.962637496396672	0.03\\
0.962637496396672	0.04\\
1.05244541067206	0.04\\
1.05244541067206	0.05\\
1.10250143007319	0.05\\
1.10250143007319	0.06\\
1.17886783581299	0.06\\
1.17886783581299	0.07\\
1.35495039139977	0.07\\
1.35495039139977	0.08\\
1.43622581395633	0.08\\
1.43622581395633	0.09\\
1.59532882252705	0.09\\
1.59532882252705	0.1\\
1.60607228785363	0.1\\
1.60607228785363	0.11\\
1.63840176537986	0.11\\
1.63840176537986	0.12\\
1.67001794041246	0.12\\
1.67001794041246	0.13\\
1.79776785287259	0.13\\
1.79776785287259	0.14\\
1.87308255840645	0.14\\
1.87308255840645	0.15\\
1.8921318094394	0.15\\
1.8921318094394	0.16\\
1.91526766784396	0.16\\
1.91526766784396	0.17\\
2.11236555779559	0.17\\
2.11236555779559	0.18\\
2.17766091846234	0.18\\
2.17766091846234	0.19\\
2.26986582685313	0.19\\
2.26986582685313	0.2\\
2.30460044663294	0.2\\
2.30460044663294	0.21\\
2.32290708087325	0.21\\
2.32290708087325	0.22\\
2.38075839272949	0.22\\
2.38075839272949	0.23\\
2.41578770511005	0.23\\
2.41578770511005	0.24\\
2.46251474794529	0.24\\
2.46251474794529	0.25\\
2.52319163997132	0.25\\
2.52319163997132	0.26\\
2.6301447125439	0.26\\
2.6301447125439	0.27\\
2.65891881578562	0.27\\
2.65891881578562	0.28\\
2.72513206298703	0.28\\
2.72513206298703	0.29\\
2.77837382310427	0.29\\
2.77837382310427	0.3\\
2.86508429292863	0.3\\
2.86508429292863	0.31\\
2.89536910349798	0.31\\
2.89536910349798	0.32\\
2.91002104714782	0.32\\
2.91002104714782	0.33\\
2.96153546097308	0.33\\
2.96153546097308	0.34\\
2.97994500235004	0.34\\
2.97994500235004	0.35\\
3.10327658865039	0.35\\
3.10327658865039	0.36\\
3.14995970959433	0.36\\
3.14995970959433	0.37\\
3.16174369718951	0.37\\
3.16174369718951	0.38\\
3.16394357757604	0.38\\
3.16394357757604	0.39\\
3.35696054689106	0.39\\
3.35696054689106	0.4\\
3.45441378744732	0.4\\
3.45441378744732	0.41\\
3.53011222096883	0.41\\
3.53011222096883	0.42\\
3.6918880618524	0.42\\
3.6918880618524	0.43\\
3.86494170164392	0.43\\
3.86494170164392	0.44\\
3.96668468310006	0.44\\
3.96668468310006	0.45\\
4.00360245472094	0.45\\
4.00360245472094	0.46\\
4.10721331782243	0.46\\
4.10721331782243	0.47\\
4.15196712528329	0.47\\
4.15196712528329	0.48\\
4.41925183770562	0.48\\
4.41925183770562	0.49\\
4.42865383591072	0.49\\
4.42865383591072	0.5\\
4.44032509605951	0.5\\
4.44032509605951	0.51\\
4.56959094614672	0.51\\
4.56959094614672	0.52\\
4.70468004986865	0.52\\
4.70468004986865	0.53\\
4.7262378028129	0.53\\
4.7262378028129	0.54\\
4.80932429693839	0.54\\
4.80932429693839	0.55\\
4.87708685933475	0.55\\
4.87708685933475	0.56\\
4.92299592274246	0.56\\
4.92299592274246	0.57\\
4.92997691845963	0.57\\
4.92997691845963	0.58\\
4.96583382491892	0.58\\
4.96583382491892	0.59\\
5.01271250717659	0.59\\
5.01271250717659	0.6\\
5.06948623727852	0.6\\
5.06948623727852	0.61\\
5.10818346631337	0.61\\
5.10818346631337	0.62\\
5.20473414043559	0.62\\
5.20473414043559	0.63\\
5.32525954679999	0.63\\
5.32525954679999	0.64\\
5.3349571109242	0.64\\
5.3349571109242	0.65\\
5.51182783180533	0.65\\
5.51182783180533	0.66\\
5.51183976495137	0.66\\
5.51183976495137	0.67\\
5.760952819239	0.67\\
5.760952819239	0.68\\
5.79093261988539	0.68\\
5.79093261988539	0.69\\
5.82469905577989	0.69\\
5.82469905577989	0.7\\
5.96124705133122	0.7\\
5.96124705133122	0.71\\
6.07649155240421	0.71\\
6.07649155240421	0.72\\
6.70984749275623	0.72\\
6.70984749275623	0.73\\
6.7726200676139	0.73\\
6.7726200676139	0.74\\
7.08727151518809	0.74\\
7.08727151518809	0.75\\
7.11651435134652	0.75\\
7.11651435134652	0.76\\
7.4226086304377	0.76\\
7.4226086304377	0.77\\
7.42901716322424	0.77\\
7.42901716322424	0.78\\
7.55086976788723	0.78\\
7.55086976788723	0.79\\
7.62978486068165	0.79\\
7.62978486068165	0.8\\
7.71940322949571	0.8\\
7.71940322949571	0.81\\
7.84269325521621	0.81\\
7.84269325521621	0.82\\
7.8713118667785	0.82\\
7.8713118667785	0.83\\
7.93304917719494	0.83\\
7.93304917719494	0.84\\
8.0779997868116	0.84\\
8.0779997868116	0.85\\
8.92129939503617	0.85\\
8.92129939503617	0.86\\
9.16834144772185	0.86\\
9.16834144772185	0.87\\
9.27428512665639	0.87\\
9.27428512665639	0.88\\
9.56750815079253	0.88\\
9.56750815079253	0.89\\
9.77134815728067	0.89\\
9.77134815728067	0.9\\
10.0792156904944	0.9\\
10.0792156904944	0.91\\
10.5172916811984	0.91\\
10.5172916811984	0.92\\
10.7441051449752	0.92\\
10.7441051449752	0.93\\
10.9655766183881	0.93\\
10.9655766183881	0.94\\
11.5934733768203	0.94\\
11.5934733768203	0.95\\
14.1270707173604	0.95\\
14.1270707173604	0.96\\
14.4221758676557	0.96\\
14.4221758676557	0.97\\
14.7798897848416	0.97\\
14.7798897848416	0.98\\
15.084042181962	0.98\\
15.084042181962	0.99\\
21.8554008418193	0.99\\
21.8554008418193	1\\
inf	1\\
};
%\addlegendentry{4}

\addplot [color=mycolor4,line width=0.7pt]
  table[row sep=crcr]{%
-inf	0\\
0.342586276585623	0\\
0.342586276585623	0.01\\
0.343135896019444	0.01\\
0.343135896019444	0.02\\
0.374162014715901	0.02\\
0.374162014715901	0.03\\
0.588494555704305	0.03\\
0.588494555704305	0.04\\
0.683057663327285	0.04\\
0.683057663327285	0.05\\
0.685814427272957	0.05\\
0.685814427272957	0.06\\
0.793340034604938	0.06\\
0.793340034604938	0.07\\
0.964784325781854	0.07\\
0.964784325781854	0.08\\
1.00195489107026	0.08\\
1.00195489107026	0.09\\
1.04131494063076	0.09\\
1.04131494063076	0.1\\
1.05562748667803	0.1\\
1.05562748667803	0.11\\
1.07415254888446	0.11\\
1.07415254888446	0.12\\
1.09350644090186	0.12\\
1.09350644090186	0.13\\
1.16170639467124	0.13\\
1.16170639467124	0.14\\
1.23201173677105	0.14\\
1.23201173677105	0.15\\
1.24605933621793	0.15\\
1.24605933621793	0.16\\
1.26548666248968	0.16\\
1.26548666248968	0.17\\
1.28679403683184	0.17\\
1.28679403683184	0.18\\
1.30534343927848	0.18\\
1.30534343927848	0.19\\
1.34207556827777	0.19\\
1.34207556827777	0.2\\
1.3856872397158	0.2\\
1.3856872397158	0.21\\
1.39719814397281	0.21\\
1.39719814397281	0.22\\
1.52494077408284	0.22\\
1.52494077408284	0.23\\
1.59432268940082	0.23\\
1.59432268940082	0.24\\
1.6943280204275	0.24\\
1.6943280204275	0.25\\
1.69488774379654	0.25\\
1.69488774379654	0.26\\
1.72140735851882	0.26\\
1.72140735851882	0.27\\
1.72187161927743	0.27\\
1.72187161927743	0.28\\
1.74608750476928	0.28\\
1.74608750476928	0.29\\
1.75788889141196	0.29\\
1.75788889141196	0.3\\
1.81570790518844	0.3\\
1.81570790518844	0.31\\
1.84018804887365	0.31\\
1.84018804887365	0.32\\
1.85594660676089	0.32\\
1.85594660676089	0.33\\
1.87285126320892	0.33\\
1.87285126320892	0.34\\
1.92374634269441	0.34\\
1.92374634269441	0.35\\
1.93229799724717	0.35\\
1.93229799724717	0.36\\
1.93929830806015	0.36\\
1.93929830806015	0.37\\
1.94329742507661	0.37\\
1.94329742507661	0.38\\
1.9522198709119	0.38\\
1.9522198709119	0.39\\
2.05029918034104	0.39\\
2.05029918034104	0.4\\
2.0552848971113	0.4\\
2.0552848971113	0.41\\
2.10178541928592	0.41\\
2.10178541928592	0.42\\
2.2500492167747	0.42\\
2.2500492167747	0.43\\
2.26419060215028	0.43\\
2.26419060215028	0.44\\
2.31144386040701	0.44\\
2.31144386040701	0.45\\
2.37955560171042	0.45\\
2.37955560171042	0.46\\
2.3981316797291	0.46\\
2.3981316797291	0.47\\
2.43336532709309	0.47\\
2.43336532709309	0.48\\
2.4435812309165	0.48\\
2.4435812309165	0.49\\
2.51509449219593	0.49\\
2.51509449219593	0.5\\
2.53851775235032	0.5\\
2.53851775235032	0.51\\
2.54527222676502	0.51\\
2.54527222676502	0.52\\
2.61935995601678	0.52\\
2.61935995601678	0.53\\
2.66528499089509	0.53\\
2.66528499089509	0.54\\
2.70241338886572	0.54\\
2.70241338886572	0.55\\
2.81296942394702	0.55\\
2.81296942394702	0.56\\
2.83552059651467	0.56\\
2.83552059651467	0.57\\
2.87648063252775	0.57\\
2.87648063252775	0.58\\
2.92985559632977	0.58\\
2.92985559632977	0.59\\
2.94554786892853	0.59\\
2.94554786892853	0.6\\
3.17521431323013	0.6\\
3.17521431323013	0.61\\
3.28111022719309	0.61\\
3.28111022719309	0.62\\
3.37461171979324	0.62\\
3.37461171979324	0.63\\
3.58447090715015	0.63\\
3.58447090715015	0.64\\
3.63177464005655	0.64\\
3.63177464005655	0.65\\
3.77218335251731	0.65\\
3.77218335251731	0.66\\
4.02332732671485	0.66\\
4.02332732671485	0.67\\
4.08288866062072	0.67\\
4.08288866062072	0.68\\
4.5368387146088	0.68\\
4.5368387146088	0.69\\
5.01295644676323	0.69\\
5.01295644676323	0.7\\
5.27869789722959	0.7\\
5.27869789722959	0.71\\
5.38952805856972	0.71\\
5.38952805856972	0.72\\
5.49703675558805	0.72\\
5.49703675558805	0.73\\
5.58443608105796	0.73\\
5.58443608105796	0.74\\
5.61112726624986	0.74\\
5.61112726624986	0.75\\
5.69908427478036	0.75\\
5.69908427478036	0.76\\
5.78309146235517	0.76\\
5.78309146235517	0.77\\
5.78721522359392	0.77\\
5.78721522359392	0.78\\
5.80510580139341	0.78\\
5.80510580139341	0.79\\
6.4357583692877	0.79\\
6.4357583692877	0.8\\
6.43594486479588	0.8\\
6.43594486479588	0.81\\
6.67912029093863	0.81\\
6.67912029093863	0.82\\
6.73065563765089	0.82\\
6.73065563765089	0.83\\
7.02780325363336	0.83\\
7.02780325363336	0.84\\
7.24494725373788	0.84\\
7.24494725373788	0.85\\
7.86391095106167	0.85\\
7.86391095106167	0.86\\
7.90851009993644	0.86\\
7.90851009993644	0.87\\
7.95796830226094	0.87\\
7.95796830226094	0.88\\
7.98186409135726	0.88\\
7.98186409135726	0.89\\
8.0591841400492	0.89\\
8.0591841400492	0.9\\
8.37153439160932	0.9\\
8.37153439160932	0.91\\
8.91634482075722	0.91\\
8.91634482075722	0.92\\
8.97748139377842	0.92\\
8.97748139377842	0.93\\
10.1851615645953	0.93\\
10.1851615645953	0.94\\
11.007799261025	0.94\\
11.007799261025	0.95\\
12.2912334723447	0.95\\
12.2912334723447	0.96\\
12.6894665315208	0.96\\
12.6894665315208	0.97\\
13.6905882227083	0.97\\
13.6905882227083	0.98\\
14.9966715610429	0.98\\
14.9966715610429	0.99\\
15.3599477236881	0.99\\
15.3599477236881	1\\
inf	1\\
};
%\addlegendentry{8}

\addplot [color=mycolor5,line width=0.7pt]
  table[row sep=crcr]{%
-inf	0\\
0.171083159954096	0\\
0.171083159954096	0.01\\
0.269074125742399	0.01\\
0.269074125742399	0.02\\
0.304962675295356	0.02\\
0.304962675295356	0.03\\
0.327171396728276	0.03\\
0.327171396728276	0.04\\
0.367267148651973	0.04\\
0.367267148651973	0.05\\
0.422272407616859	0.05\\
0.422272407616859	0.06\\
0.441731382183302	0.06\\
0.441731382183302	0.07\\
0.46127985832897	0.07\\
0.46127985832897	0.08\\
0.463604498652339	0.08\\
0.463604498652339	0.09\\
0.479072033801952	0.09\\
0.479072033801952	0.1\\
0.507524011543438	0.1\\
0.507524011543438	0.11\\
0.512733842649553	0.11\\
0.512733842649553	0.12\\
0.51753062018917	0.12\\
0.51753062018917	0.13\\
0.560038852736603	0.13\\
0.560038852736603	0.14\\
0.57389617391351	0.14\\
0.57389617391351	0.15\\
0.599284495456107	0.15\\
0.599284495456107	0.16\\
0.601845616859323	0.16\\
0.601845616859323	0.17\\
0.60435297684392	0.17\\
0.60435297684392	0.18\\
0.651328021905207	0.18\\
0.651328021905207	0.19\\
0.668145373587269	0.19\\
0.668145373587269	0.2\\
0.739410505408412	0.2\\
0.739410505408412	0.21\\
0.764827163192505	0.21\\
0.764827163192505	0.22\\
0.774088758355054	0.22\\
0.774088758355054	0.23\\
0.783356435998398	0.23\\
0.783356435998398	0.24\\
0.807049185657336	0.24\\
0.807049185657336	0.25\\
0.810860815055677	0.25\\
0.810860815055677	0.26\\
0.833819824781039	0.26\\
0.833819824781039	0.27\\
0.834312750797963	0.27\\
0.834312750797963	0.28\\
0.84190240336333	0.28\\
0.84190240336333	0.29\\
0.865265137140852	0.29\\
0.865265137140852	0.3\\
0.901656479241935	0.3\\
0.901656479241935	0.31\\
0.928549363636289	0.31\\
0.928549363636289	0.32\\
0.929913244354853	0.32\\
0.929913244354853	0.33\\
0.930094712298193	0.33\\
0.930094712298193	0.34\\
0.937936587073442	0.34\\
0.937936587073442	0.35\\
0.950517833002408	0.35\\
0.950517833002408	0.36\\
0.978455408058639	0.36\\
0.978455408058639	0.37\\
0.980715488227303	0.37\\
0.980715488227303	0.38\\
1.08890438324673	0.38\\
1.08890438324673	0.39\\
1.11771175822426	0.39\\
1.11771175822426	0.4\\
1.13439624232686	0.4\\
1.13439624232686	0.41\\
1.14587422295932	0.41\\
1.14587422295932	0.42\\
1.17851111609263	0.42\\
1.17851111609263	0.43\\
1.20695893925848	0.43\\
1.20695893925848	0.44\\
1.22275979925426	0.44\\
1.22275979925426	0.45\\
1.2578348847427	0.45\\
1.2578348847427	0.46\\
1.27308890321539	0.46\\
1.27308890321539	0.47\\
1.27424096076476	0.47\\
1.27424096076476	0.48\\
1.27436927657064	0.48\\
1.27436927657064	0.49\\
1.29362344510801	0.49\\
1.29362344510801	0.5\\
1.30454853560738	0.5\\
1.30454853560738	0.51\\
1.30557814338625	0.51\\
1.30557814338625	0.52\\
1.32737651838362	0.52\\
1.32737651838362	0.53\\
1.32762283933603	0.53\\
1.32762283933603	0.54\\
1.33084262520001	0.54\\
1.33084262520001	0.55\\
1.34595824563923	0.55\\
1.34595824563923	0.56\\
1.40387333180955	0.56\\
1.40387333180955	0.57\\
1.45553940593438	0.57\\
1.45553940593438	0.58\\
1.49203434665814	0.58\\
1.49203434665814	0.59\\
1.54933680865438	0.59\\
1.54933680865438	0.6\\
1.54996398684305	0.6\\
1.54996398684305	0.61\\
1.61831877922148	0.61\\
1.61831877922148	0.62\\
1.64882375107856	0.62\\
1.64882375107856	0.63\\
1.65259028578859	0.63\\
1.65259028578859	0.64\\
1.66350277270553	0.64\\
1.66350277270553	0.65\\
1.7026150466707	0.65\\
1.7026150466707	0.66\\
1.77652918611845	0.66\\
1.77652918611845	0.67\\
1.83171483674219	0.67\\
1.83171483674219	0.68\\
1.96819715017838	0.68\\
1.96819715017838	0.69\\
2.01923647215902	0.69\\
2.01923647215902	0.7\\
2.03895590169786	0.7\\
2.03895590169786	0.71\\
2.07445893156101	0.71\\
2.07445893156101	0.72\\
2.12995125882073	0.72\\
2.12995125882073	0.73\\
2.16303914342571	0.73\\
2.16303914342571	0.74\\
2.28575515869577	0.74\\
2.28575515869577	0.75\\
2.31171520961554	0.75\\
2.31171520961554	0.76\\
2.32843695563986	0.76\\
2.32843695563986	0.77\\
2.42608884721038	0.77\\
2.42608884721038	0.78\\
2.52038112534964	0.78\\
2.52038112534964	0.79\\
2.61980921484897	0.79\\
2.61980921484897	0.8\\
2.65237946934835	0.8\\
2.65237946934835	0.81\\
2.67194943396037	0.81\\
2.67194943396037	0.82\\
2.8536372028752	0.82\\
2.8536372028752	0.83\\
3.31440859751607	0.83\\
3.31440859751607	0.84\\
3.31958578904136	0.84\\
3.31958578904136	0.85\\
3.3378045624148	0.85\\
3.3378045624148	0.86\\
3.38272899077345	0.86\\
3.38272899077345	0.87\\
3.76445951530957	0.87\\
3.76445951530957	0.88\\
4.0250314986286	0.88\\
4.0250314986286	0.89\\
4.05704840824982	0.89\\
4.05704840824982	0.9\\
4.15327308587579	0.9\\
4.15327308587579	0.91\\
4.42606506975387	0.91\\
4.42606506975387	0.92\\
4.53913505189049	0.92\\
4.53913505189049	0.93\\
5.06689369795091	0.93\\
5.06689369795091	0.94\\
5.82683640364	0.94\\
5.82683640364	0.95\\
6.39070791961791	0.95\\
6.39070791961791	0.96\\
6.68727644481411	0.96\\
6.68727644481411	0.97\\
6.97088009229216	0.97\\
6.97088009229216	0.98\\
11.7593360495509	0.98\\
11.7593360495509	0.99\\
14.4962896588754	0.99\\
14.4962896588754	1\\
inf	1\\
};
%\addlegendentry{16}

\addplot [color=mycolor6,line width=0.7pt]
  table[row sep=crcr]{%
-inf	0\\
0.138728582749421	0\\
0.138728582749421	0.01\\
0.214108644375846	0.01\\
0.214108644375846	0.02\\
0.241181363162243	0.02\\
0.241181363162243	0.03\\
0.248017887355008	0.03\\
0.248017887355008	0.04\\
0.266134884158181	0.04\\
0.266134884158181	0.05\\
0.28326213028894	0.05\\
0.28326213028894	0.06\\
0.305299741462719	0.06\\
0.305299741462719	0.07\\
0.331449396029825	0.07\\
0.331449396029825	0.08\\
0.334509327624709	0.08\\
0.334509327624709	0.09\\
0.355664518854367	0.09\\
0.355664518854367	0.1\\
0.382848381802336	0.1\\
0.382848381802336	0.11\\
0.390644650073286	0.11\\
0.390644650073286	0.12\\
0.397609576648951	0.12\\
0.397609576648951	0.13\\
0.401958670259519	0.13\\
0.401958670259519	0.14\\
0.408729817722103	0.14\\
0.408729817722103	0.15\\
0.414016910096915	0.15\\
0.414016910096915	0.16\\
0.446803224704251	0.16\\
0.446803224704251	0.17\\
0.467612829012685	0.17\\
0.467612829012685	0.18\\
0.470030081025521	0.18\\
0.470030081025521	0.19\\
0.470892242919728	0.19\\
0.470892242919728	0.2\\
0.478897043233896	0.2\\
0.478897043233896	0.21\\
0.481986639752626	0.21\\
0.481986639752626	0.22\\
0.489979090581788	0.22\\
0.489979090581788	0.23\\
0.511701829566543	0.23\\
0.511701829566543	0.24\\
0.517162926740584	0.24\\
0.517162926740584	0.25\\
0.51952054009375	0.25\\
0.51952054009375	0.26\\
0.521437568601479	0.26\\
0.521437568601479	0.27\\
0.534239955749921	0.27\\
0.534239955749921	0.28\\
0.537410190228055	0.28\\
0.537410190228055	0.29\\
0.539818926721574	0.29\\
0.539818926721574	0.3\\
0.545524414076764	0.3\\
0.545524414076764	0.31\\
0.546439526846052	0.31\\
0.546439526846052	0.32\\
0.549027732858334	0.32\\
0.549027732858334	0.33\\
0.550868770576643	0.33\\
0.550868770576643	0.34\\
0.568098057458094	0.34\\
0.568098057458094	0.35\\
0.591687377784511	0.35\\
0.591687377784511	0.36\\
0.617394886584491	0.36\\
0.617394886584491	0.37\\
0.621170293416942	0.37\\
0.621170293416942	0.38\\
0.632635416390626	0.38\\
0.632635416390626	0.39\\
0.640083685772302	0.39\\
0.640083685772302	0.4\\
0.659614940449425	0.4\\
0.659614940449425	0.41\\
0.734384259908133	0.41\\
0.734384259908133	0.42\\
0.772375029177168	0.42\\
0.772375029177168	0.43\\
0.772748621979217	0.43\\
0.772748621979217	0.44\\
0.775246103197171	0.44\\
0.775246103197171	0.45\\
0.78629108672022	0.45\\
0.78629108672022	0.46\\
0.786547784716161	0.46\\
0.786547784716161	0.47\\
0.815709023041874	0.47\\
0.815709023041874	0.48\\
0.824535731793223	0.48\\
0.824535731793223	0.49\\
0.841664002204177	0.49\\
0.841664002204177	0.5\\
0.858398230011372	0.5\\
0.858398230011372	0.51\\
0.868075906054955	0.51\\
0.868075906054955	0.52\\
0.880463517350824	0.52\\
0.880463517350824	0.53\\
0.90342484440906	0.53\\
0.90342484440906	0.54\\
0.922687650503285	0.54\\
0.922687650503285	0.55\\
0.934995935629238	0.55\\
0.934995935629238	0.56\\
0.985995399782663	0.56\\
0.985995399782663	0.57\\
1.0463673124549	0.57\\
1.0463673124549	0.58\\
1.0675976917235	0.58\\
1.0675976917235	0.59\\
1.06956406582572	0.59\\
1.06956406582572	0.6\\
1.08712898180299	0.6\\
1.08712898180299	0.61\\
1.15674163885869	0.61\\
1.15674163885869	0.62\\
1.16926837422808	0.62\\
1.16926837422808	0.63\\
1.1943490403741	0.63\\
1.1943490403741	0.64\\
1.21156649784564	0.64\\
1.21156649784564	0.65\\
1.21699547786728	0.65\\
1.21699547786728	0.66\\
1.24225850045148	0.66\\
1.24225850045148	0.67\\
1.24626557048581	0.67\\
1.24626557048581	0.68\\
1.29199955743738	0.68\\
1.29199955743738	0.69\\
1.35950724101207	0.69\\
1.35950724101207	0.7\\
1.38190308638362	0.7\\
1.38190308638362	0.71\\
1.40114198091403	0.71\\
1.40114198091403	0.72\\
1.41955435337	0.72\\
1.41955435337	0.73\\
1.56065174588531	0.73\\
1.56065174588531	0.74\\
1.57243225458265	0.74\\
1.57243225458265	0.75\\
1.59382103676186	0.75\\
1.59382103676186	0.76\\
1.60947117988551	0.76\\
1.60947117988551	0.77\\
1.6097926274142	0.77\\
1.6097926274142	0.78\\
1.62139476657625	0.78\\
1.62139476657625	0.79\\
1.64948551746088	0.79\\
1.64948551746088	0.8\\
1.68434689257578	0.8\\
1.68434689257578	0.81\\
1.72852699953126	0.81\\
1.72852699953126	0.82\\
1.79998765299133	0.82\\
1.79998765299133	0.83\\
1.81825071007973	0.83\\
1.81825071007973	0.84\\
1.89727944244396	0.84\\
1.89727944244396	0.85\\
1.9220093832347	0.85\\
1.9220093832347	0.86\\
1.98078436043484	0.86\\
1.98078436043484	0.87\\
2.12221822784844	0.87\\
2.12221822784844	0.88\\
2.18580518201026	0.88\\
2.18580518201026	0.89\\
2.19986504464784	0.89\\
2.19986504464784	0.9\\
2.32276492200036	0.9\\
2.32276492200036	0.91\\
2.38877657172804	0.91\\
2.38877657172804	0.92\\
2.43006222038217	0.92\\
2.43006222038217	0.93\\
2.91748169439551	0.93\\
2.91748169439551	0.94\\
3.088298757225	0.94\\
3.088298757225	0.95\\
3.23202185182057	0.95\\
3.23202185182057	0.96\\
3.34924274227058	0.96\\
3.34924274227058	0.97\\
4.31368640505699	0.97\\
4.31368640505699	0.98\\
4.36611107878339	0.98\\
4.36611107878339	0.99\\
5.64259482061546	0.99\\
5.64259482061546	1\\
inf	1\\
};
%\addlegendentry{32}

\addplot [color=mycolor7,line width=0.7pt]
  table[row sep=crcr]{%
-inf	0\\
0.131825703133813	0\\
0.131825703133813	0.01\\
0.215559556458991	0.01\\
0.215559556458991	0.02\\
0.2173899717775	0.02\\
0.2173899717775	0.03\\
0.217480322837405	0.03\\
0.217480322837405	0.04\\
0.218745556053794	0.04\\
0.218745556053794	0.05\\
0.220302987162224	0.05\\
0.220302987162224	0.06\\
0.246746578748743	0.06\\
0.246746578748743	0.07\\
0.250263429429822	0.07\\
0.250263429429822	0.08\\
0.261016373780285	0.08\\
0.261016373780285	0.09\\
0.266240014417202	0.09\\
0.266240014417202	0.1\\
0.268678935546516	0.1\\
0.268678935546516	0.11\\
0.269731140970187	0.11\\
0.269731140970187	0.12\\
0.269886021213268	0.12\\
0.269886021213268	0.13\\
0.274016309297771	0.13\\
0.274016309297771	0.14\\
0.287535771890066	0.14\\
0.287535771890066	0.15\\
0.290505728974326	0.15\\
0.290505728974326	0.16\\
0.310863735989227	0.16\\
0.310863735989227	0.17\\
0.321053638533792	0.17\\
0.321053638533792	0.18\\
0.322680842381877	0.18\\
0.322680842381877	0.19\\
0.333203266287217	0.19\\
0.333203266287217	0.2\\
0.338173159188376	0.2\\
0.338173159188376	0.21\\
0.346953307322463	0.21\\
0.346953307322463	0.22\\
0.348557748125636	0.22\\
0.348557748125636	0.23\\
0.349858910349437	0.23\\
0.349858910349437	0.24\\
0.354711194910359	0.24\\
0.354711194910359	0.25\\
0.366819231719452	0.25\\
0.366819231719452	0.26\\
0.372981932447435	0.26\\
0.372981932447435	0.27\\
0.378809729957063	0.27\\
0.378809729957063	0.28\\
0.379546750601405	0.28\\
0.379546750601405	0.29\\
0.38644166489161	0.29\\
0.38644166489161	0.3\\
0.38926226304798	0.3\\
0.38926226304798	0.31\\
0.393773262448461	0.31\\
0.393773262448461	0.32\\
0.401318166266583	0.32\\
0.401318166266583	0.33\\
0.402308850831875	0.33\\
0.402308850831875	0.34\\
0.405267817735552	0.34\\
0.405267817735552	0.35\\
0.417915958036944	0.35\\
0.417915958036944	0.36\\
0.418919124630743	0.36\\
0.418919124630743	0.37\\
0.421974665593938	0.37\\
0.421974665593938	0.38\\
0.42499985064181	0.38\\
0.42499985064181	0.39\\
0.432004976043425	0.39\\
0.432004976043425	0.4\\
0.432556982062051	0.4\\
0.432556982062051	0.41\\
0.443976269161948	0.41\\
0.443976269161948	0.42\\
0.461011228317488	0.42\\
0.461011228317488	0.43\\
0.468144159207378	0.43\\
0.468144159207378	0.44\\
0.470590509773427	0.44\\
0.470590509773427	0.45\\
0.487478650734358	0.45\\
0.487478650734358	0.46\\
0.495187067137442	0.46\\
0.495187067137442	0.47\\
0.496263509162821	0.47\\
0.496263509162821	0.48\\
0.509318390573494	0.48\\
0.509318390573494	0.49\\
0.522073380409752	0.49\\
0.522073380409752	0.5\\
0.543550184511435	0.5\\
0.543550184511435	0.51\\
0.544706816904464	0.51\\
0.544706816904464	0.52\\
0.551382763434028	0.52\\
0.551382763434028	0.53\\
0.560940080282311	0.53\\
0.560940080282311	0.54\\
0.575094449665532	0.54\\
0.575094449665532	0.55\\
0.590315912468696	0.55\\
0.590315912468696	0.56\\
0.603796639989148	0.56\\
0.603796639989148	0.57\\
0.605475784555797	0.57\\
0.605475784555797	0.58\\
0.608876379412595	0.58\\
0.608876379412595	0.59\\
0.610323392047237	0.59\\
0.610323392047237	0.6\\
0.655093158025327	0.6\\
0.655093158025327	0.61\\
0.662110309641281	0.61\\
0.662110309641281	0.62\\
0.663073099123561	0.62\\
0.663073099123561	0.63\\
0.675183272336235	0.63\\
0.675183272336235	0.64\\
0.681441412694866	0.64\\
0.681441412694866	0.65\\
0.701312876660268	0.65\\
0.701312876660268	0.66\\
0.716648241197727	0.66\\
0.716648241197727	0.67\\
0.719526372657106	0.67\\
0.719526372657106	0.68\\
0.725162773114295	0.68\\
0.725162773114295	0.69\\
0.7341910710195	0.69\\
0.7341910710195	0.7\\
0.735751285880322	0.7\\
0.735751285880322	0.71\\
0.740865775671163	0.71\\
0.740865775671163	0.72\\
0.7409225253357	0.72\\
0.7409225253357	0.73\\
0.750240183415842	0.73\\
0.750240183415842	0.74\\
0.761746819558183	0.74\\
0.761746819558183	0.75\\
0.792947212370173	0.75\\
0.792947212370173	0.76\\
0.801209998978462	0.76\\
0.801209998978462	0.77\\
0.815652074350861	0.77\\
0.815652074350861	0.78\\
0.845176993602576	0.78\\
0.845176993602576	0.79\\
0.845579903671659	0.79\\
0.845579903671659	0.8\\
0.911965917500925	0.8\\
0.911965917500925	0.81\\
0.943821671427009	0.81\\
0.943821671427009	0.82\\
0.945626187572482	0.82\\
0.945626187572482	0.83\\
0.965527975210007	0.83\\
0.965527975210007	0.84\\
0.970458113271487	0.84\\
0.970458113271487	0.85\\
1.06237422951722	0.85\\
1.06237422951722	0.86\\
1.1145237321321	0.86\\
1.1145237321321	0.87\\
1.11718402584333	0.87\\
1.11718402584333	0.88\\
1.1908719524492	0.88\\
1.1908719524492	0.89\\
1.27712001891137	0.89\\
1.27712001891137	0.9\\
1.27992067955549	0.9\\
1.27992067955549	0.91\\
1.33913471320578	0.91\\
1.33913471320578	0.92\\
1.34859280408818	0.92\\
1.34859280408818	0.93\\
1.35423700670951	0.93\\
1.35423700670951	0.94\\
1.56538689382909	0.94\\
1.56538689382909	0.95\\
1.63279962577734	0.95\\
1.63279962577734	0.96\\
1.63537607273343	0.96\\
1.63537607273343	0.97\\
2.29952081047072	0.97\\
2.29952081047072	0.98\\
5.74042312503304	0.98\\
5.74042312503304	0.99\\
363.398405477968	0.99\\
363.398405477968	1\\
inf	1\\
};
%\addlegendentry{64}

\addplot [color=mycolor2, forget plot, line width=0.7pt]%only marks, scatter, mark=x, mark size=2.9pt]
  table[row sep=crcr]{%
-inf	0\\
0.120403140900258	0\\
0.120403140900258	0.01\\
0.146452971373464	0.01\\
0.146452971373464	0.02\\
0.155230699009699	0.02\\
0.155230699009699	0.03\\
0.163114797798653	0.03\\
0.163114797798653	0.04\\
0.177357594103634	0.04\\
0.177357594103634	0.05\\
0.179162235123654	0.05\\
0.179162235123654	0.06\\
0.179606004968496	0.06\\
0.179606004968496	0.07\\
0.195428237359449	0.07\\
0.195428237359449	0.08\\
0.20308517344866	0.08\\
0.20308517344866	0.09\\
0.20703367277857	0.09\\
0.20703367277857	0.1\\
0.21082518066634	0.1\\
0.21082518066634	0.11\\
0.211093493769465	0.11\\
0.211093493769465	0.12\\
0.213095608813526	0.12\\
0.213095608813526	0.13\\
0.220021746671141	0.13\\
0.220021746671141	0.14\\
0.223507567068835	0.14\\
0.223507567068835	0.15\\
0.232609634448675	0.15\\
0.232609634448675	0.16\\
0.237580478604198	0.16\\
0.237580478604198	0.17\\
0.244279787141878	0.17\\
0.244279787141878	0.18\\
0.246977316925599	0.18\\
0.246977316925599	0.19\\
0.251416507338451	0.19\\
0.251416507338451	0.2\\
0.255111597664276	0.2\\
0.255111597664276	0.21\\
0.258073891577459	0.21\\
0.258073891577459	0.22\\
0.25833218543636	0.22\\
0.25833218543636	0.23\\
0.261370640167871	0.23\\
0.261370640167871	0.24\\
0.262484732333105	0.24\\
0.262484732333105	0.25\\
0.270315204615358	0.25\\
0.270315204615358	0.26\\
0.271006492887933	0.26\\
0.271006492887933	0.27\\
0.275676978425068	0.27\\
0.275676978425068	0.28\\
0.276399819723849	0.28\\
0.276399819723849	0.29\\
0.277362643079256	0.29\\
0.277362643079256	0.3\\
0.281017876122247	0.3\\
0.281017876122247	0.31\\
0.28285460688499	0.31\\
0.28285460688499	0.32\\
0.286243040552937	0.32\\
0.286243040552937	0.33\\
0.289132348352341	0.33\\
0.289132348352341	0.34\\
0.289712225235786	0.34\\
0.289712225235786	0.35\\
0.292707190983197	0.35\\
0.292707190983197	0.36\\
0.292962201490121	0.36\\
0.292962201490121	0.37\\
0.294123073960695	0.37\\
0.294123073960695	0.38\\
0.296970763025873	0.38\\
0.296970763025873	0.39\\
0.297740780890334	0.39\\
0.297740780890334	0.4\\
0.298568119484819	0.4\\
0.298568119484819	0.41\\
0.300930530581829	0.41\\
0.300930530581829	0.42\\
0.303111995741053	0.42\\
0.303111995741053	0.43\\
0.307583652739438	0.43\\
0.307583652739438	0.44\\
0.309819720092151	0.44\\
0.309819720092151	0.45\\
0.31622911767991	0.45\\
0.31622911767991	0.46\\
0.32497230510695	0.46\\
0.32497230510695	0.47\\
0.327160139238012	0.47\\
0.327160139238012	0.48\\
0.32746805247275	0.48\\
0.32746805247275	0.49\\
0.33295977060791	0.49\\
0.33295977060791	0.5\\
0.338538088385428	0.5\\
0.338538088385428	0.51\\
0.339128692439553	0.51\\
0.339128692439553	0.52\\
0.339967054957063	0.52\\
0.339967054957063	0.53\\
0.343407033968161	0.53\\
0.343407033968161	0.54\\
0.346667779314529	0.54\\
0.346667779314529	0.55\\
0.352024296767169	0.55\\
0.352024296767169	0.56\\
0.363721374848745	0.56\\
0.363721374848745	0.57\\
0.379089198166214	0.57\\
0.379089198166214	0.58\\
0.388340960003313	0.58\\
0.388340960003313	0.59\\
0.393720379545225	0.59\\
0.393720379545225	0.6\\
0.394932891802	0.6\\
0.394932891802	0.61\\
0.396987944300512	0.61\\
0.396987944300512	0.62\\
0.400335254748047	0.62\\
0.400335254748047	0.63\\
0.403611821775451	0.63\\
0.403611821775451	0.64\\
0.406494683155117	0.64\\
0.406494683155117	0.65\\
0.40950278790901	0.65\\
0.40950278790901	0.66\\
0.413511482724558	0.66\\
0.413511482724558	0.67\\
0.425259377137486	0.67\\
0.425259377137486	0.68\\
0.433006895197899	0.68\\
0.433006895197899	0.69\\
0.435151583068849	0.69\\
0.435151583068849	0.7\\
0.442269922300422	0.7\\
0.442269922300422	0.71\\
0.456721392997292	0.71\\
0.456721392997292	0.72\\
0.478426086374021	0.72\\
0.478426086374021	0.73\\
0.483172333014746	0.73\\
0.483172333014746	0.74\\
0.494497829018759	0.74\\
0.494497829018759	0.75\\
0.511643259765429	0.75\\
0.511643259765429	0.76\\
0.539229829116152	0.76\\
0.539229829116152	0.77\\
0.548082774778356	0.77\\
0.548082774778356	0.78\\
0.569254699241509	0.78\\
0.569254699241509	0.79\\
0.599836088538555	0.79\\
0.599836088538555	0.8\\
0.615900982676185	0.8\\
0.615900982676185	0.81\\
0.61678220980581	0.81\\
0.61678220980581	0.82\\
0.630569647735796	0.82\\
0.630569647735796	0.83\\
0.631170663067988	0.83\\
0.631170663067988	0.84\\
0.636124294864833	0.84\\
0.636124294864833	0.85\\
0.643252957718979	0.85\\
0.643252957718979	0.86\\
0.672370221714553	0.86\\
0.672370221714553	0.87\\
0.691609439488288	0.87\\
0.691609439488288	0.88\\
0.712294575577874	0.88\\
0.712294575577874	0.89\\
0.72925140206823	0.89\\
0.72925140206823	0.9\\
0.801408278636057	0.9\\
0.801408278636057	0.91\\
0.815604271824607	0.91\\
0.815604271824607	0.92\\
1.05042750861015	0.92\\
1.05042750861015	0.93\\
1.07701199196496	0.93\\
1.07701199196496	0.94\\
1.19746750224022	0.94\\
1.19746750224022	0.95\\
1.52712222673741	0.95\\
1.52712222673741	0.96\\
1.7718576751522	0.96\\
1.7718576751522	0.97\\
564.138675271851	0.97\\
564.138675271851	0.98\\
952.713181999522	0.98\\
952.713181999522	0.99\\
1317.52754909517	0.99\\
1317.52754909517	1\\
inf	1\\
};
%\addlegendentry{2Q}

\addplot [color=mycolor1, forget plot, line width=1pt,dashed]
  table[row sep=crcr]{%
-inf	0\\
0.13090273374571	0\\
0.13090273374571	0.01\\
0.15823408798804	0.01\\
0.15823408798804	0.02\\
0.160676553791634	0.02\\
0.160676553791634	0.03\\
0.169561822607325	0.03\\
0.169561822607325	0.04\\
0.171009194562051	0.04\\
0.171009194562051	0.05\\
0.17871784782044	0.05\\
0.17871784782044	0.06\\
0.178920036320146	0.06\\
0.178920036320146	0.07\\
0.206639305318035	0.07\\
0.206639305318035	0.08\\
0.214368307663388	0.08\\
0.214368307663388	0.09\\
0.216148753033257	0.09\\
0.216148753033257	0.1\\
0.217774633882972	0.1\\
0.217774633882972	0.11\\
0.225167589505474	0.11\\
0.225167589505474	0.12\\
0.229298362479922	0.12\\
0.229298362479922	0.13\\
0.232534311476608	0.13\\
0.232534311476608	0.14\\
0.243388105431094	0.14\\
0.243388105431094	0.15\\
0.244014997012756	0.15\\
0.244014997012756	0.16\\
0.244806362490182	0.16\\
0.244806362490182	0.17\\
0.252681884981449	0.17\\
0.252681884981449	0.18\\
0.252793538087708	0.18\\
0.252793538087708	0.19\\
0.258212358843367	0.19\\
0.258212358843367	0.2\\
0.259836434863121	0.2\\
0.259836434863121	0.21\\
0.265473462137174	0.21\\
0.265473462137174	0.22\\
0.266369364102539	0.22\\
0.266369364102539	0.23\\
0.268313676719002	0.23\\
0.268313676719002	0.24\\
0.268981083242531	0.24\\
0.268981083242531	0.25\\
0.272314518532226	0.25\\
0.272314518532226	0.26\\
0.275822807110288	0.26\\
0.275822807110288	0.27\\
0.277728116604977	0.27\\
0.277728116604977	0.28\\
0.28103982130918	0.28\\
0.28103982130918	0.29\\
0.285164056221443	0.29\\
0.285164056221443	0.3\\
0.287997907049668	0.3\\
0.287997907049668	0.31\\
0.291994219055142	0.31\\
0.291994219055142	0.32\\
0.292404906531221	0.32\\
0.292404906531221	0.33\\
0.292965121626691	0.33\\
0.292965121626691	0.34\\
0.299837798359027	0.34\\
0.299837798359027	0.35\\
0.30510305362692	0.35\\
0.30510305362692	0.36\\
0.309535626017013	0.36\\
0.309535626017013	0.37\\
0.309632907430309	0.37\\
0.309632907430309	0.38\\
0.310389547231923	0.38\\
0.310389547231923	0.39\\
0.312556513624734	0.39\\
0.312556513624734	0.4\\
0.314290461708114	0.4\\
0.314290461708114	0.41\\
0.317522952755251	0.41\\
0.317522952755251	0.42\\
0.317564060080815	0.42\\
0.317564060080815	0.43\\
0.324920584375538	0.43\\
0.324920584375538	0.44\\
0.329399088717952	0.44\\
0.329399088717952	0.45\\
0.329673103404902	0.45\\
0.329673103404902	0.46\\
0.329878305993718	0.46\\
0.329878305993718	0.47\\
0.330464706233402	0.47\\
0.330464706233402	0.48\\
0.334357098323705	0.48\\
0.334357098323705	0.49\\
0.335090719680719	0.49\\
0.335090719680719	0.5\\
0.339490656601979	0.5\\
0.339490656601979	0.51\\
0.341672636656631	0.51\\
0.341672636656631	0.52\\
0.348317556739002	0.52\\
0.348317556739002	0.53\\
0.361570554074292	0.53\\
0.361570554074292	0.54\\
0.361845518836617	0.54\\
0.361845518836617	0.55\\
0.370035514930406	0.55\\
0.370035514930406	0.56\\
0.373314062605815	0.56\\
0.373314062605815	0.57\\
0.37373177068193	0.57\\
0.37373177068193	0.58\\
0.381374214161949	0.58\\
0.381374214161949	0.59\\
0.382226575737488	0.59\\
0.382226575737488	0.6\\
0.383995217788684	0.6\\
0.383995217788684	0.61\\
0.390293346859155	0.61\\
0.390293346859155	0.62\\
0.397187010356981	0.62\\
0.397187010356981	0.63\\
0.40916198108566	0.63\\
0.40916198108566	0.64\\
0.412307926418069	0.64\\
0.412307926418069	0.65\\
0.414458029259485	0.65\\
0.414458029259485	0.66\\
0.431524093846951	0.66\\
0.431524093846951	0.67\\
0.4526087740185	0.67\\
0.4526087740185	0.68\\
0.455548440907441	0.68\\
0.455548440907441	0.69\\
0.458935392420144	0.69\\
0.458935392420144	0.7\\
0.462721598900864	0.7\\
0.462721598900864	0.71\\
0.467644089595771	0.71\\
0.467644089595771	0.72\\
0.496884147101033	0.72\\
0.496884147101033	0.73\\
0.502425865601105	0.73\\
0.502425865601105	0.74\\
0.507016789685385	0.74\\
0.507016789685385	0.75\\
0.509770332915335	0.75\\
0.509770332915335	0.76\\
0.54400894901456	0.76\\
0.54400894901456	0.77\\
0.566449931193047	0.77\\
0.566449931193047	0.78\\
0.589299564416603	0.78\\
0.589299564416603	0.79\\
0.59370653253882	0.79\\
0.59370653253882	0.8\\
0.595577218787314	0.8\\
0.595577218787314	0.81\\
0.607872226718907	0.81\\
0.607872226718907	0.82\\
0.6110838380576	0.82\\
0.6110838380576	0.83\\
0.638477978818703	0.83\\
0.638477978818703	0.84\\
0.650835064058338	0.84\\
0.650835064058338	0.85\\
0.656918397479696	0.85\\
0.656918397479696	0.86\\
0.674217852023842	0.86\\
0.674217852023842	0.87\\
0.718139667375107	0.87\\
0.718139667375107	0.88\\
0.738865843067104	0.88\\
0.738865843067104	0.89\\
0.794916378497817	0.89\\
0.794916378497817	0.9\\
0.845350817788179	0.9\\
0.845350817788179	0.91\\
0.905846961815799	0.91\\
0.905846961815799	0.92\\
0.987427574151523	0.92\\
0.987427574151523	0.93\\
1.00734784986772	0.93\\
1.00734784986772	0.94\\
1.10096525586154	0.94\\
1.10096525586154	0.95\\
1.51344034857963	0.95\\
1.51344034857963	0.96\\
1.56743524327586	0.96\\
1.56743524327586	0.97\\
3.5633422263187	0.97\\
3.5633422263187	0.98\\
4.80060846512191	0.98\\
4.80060846512191	0.99\\
6.95970232870454	0.99\\
6.95970232870454	1\\
inf	1\\
};
%\addlegendentry{PEB}

\end{axis}
\node[rotate=90](ylable) at (8.1,-3.3){\footnotesize $\mathrm{Pr}(\text{Position error}<e)$};
\node[rotate=78] (BOC1) at (10.2,-3.6){\small Proposed method with $R=2$};
%\node[rotate=0,fill=white] (BOC1) at (10.5,-5){\small $\ $ };
\node[rotate=78,text=mycolor2] (BOC2) at (13.9,-4){\small $R = 2$ };
%\node[rotate=0,fill=white] (BOC1) at (11.1,-5){\small \, };
\node[rotate=70,text=mycolor3] (BOC2) at (13.2,-4){\small $R = 4$ };
%\node[rotate=0,fill=white] (BOC1) at (11.6,-5){\small \, };
\node[rotate=70,text=mycolor4] (BOC2) at (12.7,-4){\small $R = 8$ };
%\node[rotate=0,fill=white] (BOC1) at (12,-5){\small \, };
\node[rotate=70,text=mycolor5] (BOC2) at (12,-4){\small $R = 16$ };
%\node[rotate=0,fill=white] (BOC1) at (12.9,-5){\small  \ \ \ \ \ \ \ };
\node[rotate=70,text=mycolor6] (BOC2) at (11.5,-4){\small $R = 32$ };
\node[rotate=70,text=mycolor7] (BOC2) at (11,-4){\small $R = 64$ };
\node[rotate=0,fill=white] (BOC6) at (14.5,-5.5){\small Method in \cite{keykhosravi2021semi}  };
\draw (11.8,-5.5) ellipse (1.6cm and .2cm);

\end{tikzpicture}

%% file: main.bbl
% Generated by IEEEtran.bst, version: 1.14 (2015/08/26)
\begin{thebibliography}{10}
\providecommand{\url}[1]{#1}
\csname url@samestyle\endcsname
\providecommand{\newblock}{\relax}
\providecommand{\bibinfo}[2]{#2}
\providecommand{\BIBentrySTDinterwordspacing}{\spaceskip=0pt\relax}
\providecommand{\BIBentryALTinterwordstretchfactor}{4}
\providecommand{\BIBentryALTinterwordspacing}{\spaceskip=\fontdimen2\font plus
\BIBentryALTinterwordstretchfactor\fontdimen3\font minus
  \fontdimen4\font\relax}
\providecommand{\BIBforeignlanguage}[2]{{%
\expandafter\ifx\csname l@#1\endcsname\relax
\typeout{** WARNING: IEEEtran.bst: No hyphenation pattern has been}%
\typeout{** loaded for the language `#1'. Using the pattern for}%
\typeout{** the default language instead.}%
\else
\language=\csname l@#1\endcsname
\fi
#2}}
\providecommand{\BIBdecl}{\relax}
\BIBdecl

\bibitem{bjornson2021reconfigurable}
E.~Bj{\"o}rnson, H.~Wymeersch \emph{et~al.}, ``Reconfigurable intelligent
  surfaces: A signal processing perspective with wireless applications,''
  \emph{arXiv preprint arXiv:2102.00742}, 2021.

\bibitem{abrardo2020intelligent}
A.~Abrardo, D.~Dardari, and M.~Di~Renzo, ``Intelligent reflecting surfaces:
  Sum-rate optimization based on statistical {CSI},'' \emph{arXiv preprint
  arXiv:2012.10679}, 2020.

\bibitem{wu2021intelligent}
Q.~Wu, S.~Zhang \emph{et~al.}, ``Intelligent reflecting surface aided wireless
  communications: A tutorial,'' \emph{IEEE Transactions on Communications},
  2021.

\bibitem{pan2021reconfigurable}
C.~Pan, H.~Ren \emph{et~al.}, ``Reconfigurable intelligent surfaces for 6g
  systems: Principles, applications, and research directions,'' \emph{IEEE
  Commun. Magazine}, 2021.

\bibitem{di2020smart}
M.~Di~Renzo, A.~Zappone \emph{et~al.}, ``Smart radio environments empowered by
  reconfigurable intelligent surfaces: How it works, state of research, and the
  road ahead,'' \emph{IEEE J.\ Select.\ Areas \ Commun.}, vol.~38, no.~11, pp.
  2450--2525, 2020.

\bibitem{alexandropoulos2020hardware}
G.~C. Alexandropoulos and E.~Vlachos, ``A hardware architecture for
  reconfigurable intelligent surfaces with minimal active elements for explicit
  channel estimation,'' in \emph{Int. Conf. Acoust. Speech Signal Process.
  (ICASSP)}, 2020, pp. 9175--9179.

\bibitem{mishra2019channel}
D.~Mishra and H.~Johansson, ``Channel estimation and low-complexity beamforming
  design for passive intelligent surface assisted {MISO} wireless energy
  transfer,'' in \emph{Int. Conf. Acoust. Speech Signal Process. (ICASSP)},
  2019, pp. 4659--4663.

\bibitem{jensen2020optimal}
T.~L. Jensen and E.~De~Carvalho, ``An optimal channel estimation scheme for
  intelligent reflecting surfaces based on a minimum variance unbiased
  estimator,'' in \emph{Int. Conf. Acoust. Speech Signal Process. (ICASSP)},
  2020, pp. 5000--5004.

\bibitem{you2020channel}
C.~You, B.~Zheng, and R.~Zhang, ``Channel estimation and passive beamforming
  for intelligent reflecting surface: Discrete phase shift and progressive
  refinement,'' \emph{IEEE J.\ Select.\ Areas \ Commun.}, vol.~38, no.~11, pp.
  2604--2620, 2020.

\bibitem{ning2020channel}
B.~Ning, Z.~Chen \emph{et~al.}, ``Channel estimation and transmission for
  intelligent reflecting surface assisted {THz} communications,'' in \emph{IEEE
  Int. Conf. Commun. (ICC)}.\hskip 1em plus 0.5em minus 0.4em\relax IEEE, 2020.

\bibitem{alexandropoulos2020phase}
G.~C. Alexandropoulos, S.~Samarakoon \emph{et~al.}, ``Phase configuration
  learning in wireless networks with multiple reconfigurable intelligent
  surfaces,'' in \emph{IEEE Global Commun. Conf. (GLOBECOM)}, 2020.

\bibitem{wang2021joint}
W.~Wang and W.~Zhang, ``Joint beam training and positioning for intelligent
  reflecting surfaces assisted millimeter wave communications,'' \emph{IEEE
  Trans.\ Wireless Commun.}, Apr. 2021.

\bibitem{keykhosravi2021semi}
K.~Keykhosravi, M.~F. Keskin \emph{et~al.}, ``Semi-passive {3D} positioning of
  multiple {RIS}-enabled users,'' \emph{arXiv preprint arXiv:2104.12113}, 2021.

\bibitem{elzanaty2020reconfigurable}
A.~Elzanaty, A.~Guerra \emph{et~al.}, ``Reconfigurable intelligent surfaces for
  localization: Position and orientation error bounds,'' \emph{arXiv preprint
  arXiv:2009.02818}, 2020.

\bibitem{keykhosravi2020siso}
K.~Keykhosravi, M.~F. Keskin \emph{et~al.}, ``{SISO} {RIS}-enabled joint {3D}
  downlink localization and synchronization,'' in \emph{IEEE Int. Conf. Commun.
  (ICC)}, 2021.

\bibitem{lampio2020orderly}
P.~Lampio, P.~{\"O}sterg{\aa}rd, and F.~Sz{\"o}ll{\H{o}}si, ``Orderly
  generation of {Butson Hadamard} matrices,'' \emph{Mathematics of
  Computation}, vol.~89, no. 321, pp. 313--331, 2020.

\bibitem{lam2000vanishing}
T.~Y. Lam and K.~H. Leung, ``On vanishing sums of roots of unity,'' \emph{J.
  algebra}, vol. 224, no.~1, pp. 91--109, Feb. 2000.

\bibitem{laclair2016survey}
A.~J. LaClair, ``A survey on {Hadamard} matrices,'' 2016.

\bibitem{seberry1973complex}
J.~Seberry, ``Complex {Hadamard} matrices,'' 1973.

\bibitem{butson1962generalized}
A.~Butson, ``Generalized {Hadamard} matrices,'' \emph{Proceedings of the
  American Mathematical Society}, vol.~13, no.~6, pp. 894--898, 1962.

\bibitem{kay1993fundamentals}
S.~M. Kay, \emph{Fundamentals of statistical signal processing: Estimation
  Theory}.\hskip 1em plus 0.5em minus 0.4em\relax Prentice Hall PTR, 1993.

\end{thebibliography}
